\documentclass[11pt]{article}
\pdfoutput=1 
\usepackage[sc]{mathpazo}

\usepackage[margin=1in]{geometry}
\usepackage[english]{babel}
\usepackage[utf8]{inputenc}
\usepackage[compact]{titlesec}
\usepackage{cmap}
\usepackage[T1]{fontenc}
\usepackage{bm}
\usepackage{booktabs}
\usepackage{mathtools}
\usepackage{amsfonts}
\usepackage{amsmath}
\usepackage{amssymb}
\usepackage{amsthm}
\usepackage{float}
\usepackage{graphics}
\usepackage{natbib}
\usepackage[backref=page]{hyperref}
\usepackage[svgnames]{xcolor}
\hypersetup{colorlinks={true},urlcolor={blue},linkcolor={DarkBlue},citecolor=[named]{DarkGreen},linktoc=all}
\renewcommand*{\backref}[1]{}
\renewcommand*{\backrefalt}[4]{%
    \ifcase #1 (Not cited.)%
    \or        (Cited on page~#2)%
    \else      (Cited on pages~#2)%
    \fi}

\usepackage{microtype}
\usepackage[capitalise,nameinlink,noabbrev]{cleveref}
\usepackage{doi}

\newcommand{\BibTeX}{\rm B\kern-.05em{\sc i\kern-.025em b}\kern-.08em\TeX}

\usepackage[labelfont={normalfont,bf},textfont=it]{caption}
\usepackage{subcaption}

\usepackage{nicefrac}
\usepackage{pifont}

\usepackage{apxproof}
\usepackage{array,multirow,graphicx,bigdelim}

\input{insbox}
\usepackage{tikz}
\usepackage{tikz-dimline}
\usepackage{tikz-cd}
\usetikzlibrary{fit,calc,math,shapes,shapes.multipart,decorations.text,arrows,decorations.markings,decorations.pathmorphing,shapes.geometric,positioning,decorations.pathreplacing}
\tikzset{snake it/.style={decorate, decoration=snake}}

\colorlet{mygray}{gray!40}

\makeatletter
\let\oldnl\nl
\newcommand{\nonl}{\renewcommand{\nl}{\let\nl\oldnl}}
\makeatother

\usepackage{thmtools,thm-restate}
\newtheorem{definition}{Definition}
\newtheorem{lemma}{Lemma}

\Crefname{claim}{Claim}{Claims}
\Crefname{claim}{Claim}{Claims}
\Crefname{corollary}{Corollary}{Corollaries}
\Crefname{definition}{Definition}{Definitions}
\Crefname{example}{Example}{Examples}
\Crefname{lemma}{Lemma}{Lemmas}
\Crefname{property}{Property}{Properties}
\Crefname{proposition}{Proposition}{Propositions}
\Crefname{remark}{Remark}{Remarks}
\Crefname{theorem}{Theorem}{Theorems}

\allowdisplaybreaks
\newcommand{\APXH}{\textrm{\textup{APX-hard}}}
\newcommand{\NSW}{\textrm{\textup{NSW}}}
\newcommand{\eps}{\varepsilon}
\renewcommand{\hat}{\widehat}
\newcommand{\F}{\mathcal{F}}
\newcommand{\I}{\mathcal{I}}
\newcommand{\N}{\mathbb{N}}
\newcommand{\PTAS}{{\textrm{\textup{PTAS}}}}
\newcommand{\QPTAS}{{\textrm{\textup{QPTAS}}}}

\newcommand{\RainbowPerfectMatching}{{\texttt{\textup{RAINBOW PERFECT MATCHING}}}}
\newcommand{\WeightedMatching}{{\texttt{\textup{WEIGHTED PERFECT  MATCHING}}}}
\newcommand{\ExactCoverByThreeSets}{{\texttt{\textup{EXACT COVER BY THREE SETS (X3C)}}}}
\newcommand{\ThreeDPerfectMatching}{{\texttt{\textup{3-D PERFECT MATCHING}}}}
\newcommand{\ThreeSAT}{{\texttt{\textup{3-SAT}}}}
\newcommand{\MaxLinEq}{{\texttt{\textup{MAXIMUM LINEAR EQUATION}}}}
\newcommand{\Partition}{{\texttt{\textup{PARTITION}}}}
\newcommand{\Poly}{\textrm{\textup{Poly time}}}

\newcommand{\NPC}{\textrm{\textup{NP-complete}}}

\newcommand{\NPH}{\textrm{\textup{NP-hard}}}
\newcommand{\V}{\mathcal{V}}
\newcommand{\W}{\mathcal{W}}
\newcommand{\opt}{{\textup{opt}}}
\newcommand{\fpt}{\textrm{\textup{FPT}}\xspace}
\newcommand{\OO}{\ensuremath{\mathcal{O}}\xspace}
\newcommand{\capacity}{{\mathtt{cap}}\xspace}

\newcommand{\repeatcaption}[2]{%
  \renewcommand{\thefigure}{\ref{#1}}%
  \captionsetup{list=no}%
  \caption{#2 (Repeated from page \pageref{#1}.)}%
  \addtocounter{figure}{-1}
}

\newtheorem{reduction rule}{Reduction Rule}
\usepackage{colortbl}
\colorlet{myred}{red!25}
\colorlet{myblue}{blue!25}
\colorlet{mygreen}{green!25}

\usepackage[linesnumbered,lined,boxed,ruled,vlined]{algorithm2e}

\SetKwInOut{Parameters}{Parameters}
\SetKwComment{Comment}{$\triangleright$\ }{}
\SetAlFnt{\small}
\SetAlCapFnt{\small}
\SetAlCapNameFnt{\small}
\SetAlCapHSkip{0pt}
\IncMargin{-\parindent}

\title{Maximizing Nash Social Welfare under \\Two-Sided Preferences}
\author{
	\begin{tabular}{m{0.2\linewidth}m{0.3\linewidth}m{0.3\linewidth}m{0.2\linewidth}}
		\multicolumn{2}{c}{\textbf{Pallavi Jain}} & \multicolumn{2}{c}{\textbf{Rohit Vaish}}\\
		\multicolumn{2}{c}{\small{IIT Jodhpur}} & \multicolumn{2}{c}{\small{IIT Delhi}}\\
		\multicolumn{2}{c}{\href{mailto:pallavi@iitj.ac.in}{\small{\texttt{pallavi@iitj.ac.in}}}} & \multicolumn{2}{c}{\href{mailto:rvaish@iitd.ac.in}{\small{\texttt{rvaish@iitd.ac.in}}}}\\
	\end{tabular}
}

\date{}

\begin{document}

\maketitle 


\begin{abstract}
The maximum Nash social welfare (\NSW{})---which maximizes the geometric mean of agents' utilities---is a fundamental solution concept with remarkable fairness and efficiency guarantees. The computational aspects of \NSW{} have been extensively studied for \emph{one-sided} preferences where a set of agents have preferences over a set of resources. Our work deviates from this trend and studies \NSW{} maximization for \emph{two-sided} preferences, wherein a set of workers and firms, each having a cardinal valuation function, are matched with each other. We provide a systematic study of the computational complexity of maximizing \NSW{} for many-to-one matchings under two-sided preferences. Our main negative result is that maximizing \NSW{} is NP-hard even in a highly restricted setting where each firm has capacity $2$, all valuations are in the range $\{0,1,2\}$, and each agent positively values at most three other agents. In search of positive results, we develop approximation algorithms as well as parameterized algorithms in terms of natural parameters such as the number of workers, the number of firms, and the firms' capacities. We also provide algorithms for restricted domains such as symmetric binary valuations and bounded degree instances.
\end{abstract}

\section{Introduction}
\label{sec:Introduction}

The problem of matching two sets of agents, with each side having preferences over the other, has been extensively studied in economics, computer science, and artificial intelligence~\citep{RS92two,M13algorithmics,BCE+16handbook}. Such \emph{two-sided} matching problems have found several notable real-world applications such as in labor markets~\citep{RP99redesign}, school choice~\citep{AS03school}, organ exchanges~\citep{R04kidney}, 
and online dating platforms~\citep{HHA10matching}.

Despite their extensive practical applicability, commonly-used algorithms for these problems often give rise to concerns about \emph{unfairness}. A well-known example is the deferred-acceptance algorithm, which is known to strongly favor one side at the expense of the other~\citep{GS62college,MW71stable}. Similarly, the matching algorithms used by ridesharing platforms have been found to contribute to the wage gap between men and women~\citep{CDH+21gender}. 

The area of \emph{fair division} provides a formal mathematical framework for rigorously analyzing fairness concepts~\citep{BT96fair,M04fair}. The literature on fair division has focused on \emph{one-sided} preferences, wherein a set of agents have preferences over the resources. A prominent measure of fairness in this context is \emph{Nash social welfare}, defined as the geometric mean of agents' utilities~\citep{N50bargaining,KN79nash}. 

Nash welfare provides a ``sweet spot'' between the purely welfarist utilitarian objective and the purely egalitarian Rawlsian objective. It strikes a balance between the often-conflicting goals of fairness and economic efficiency, and enjoys a strong axiomatic support~\citep{M04fair,CKM+19unreasonable}. In recent years, the computational aspects of Nash welfare 
have gained considerable attention~\citep{CG18approximating,BKV18finding,ABF+21maximum,ACH+22maximizing}.


Somewhat surprisingly, Nash welfare has not been studied for \emph{two-sided} preferences. Our work aims to address this gap by proposing a systematic examination of the computational complexity of maximizing Nash welfare in the two-sided matching problem.

\subsection{Summary of Contributions}
\label{subsec:Our_Contributions}

Our model consists of two disjoint sets of agents: \emph{workers} and \emph{firms}. Each worker has a nonnegative valuation for every firm and can be matched with at most one firm. Each firm can be matched with multiple workers (up to its capacity) and has \emph{additive} valuations over the workers. The goal is to find a \emph{Nash optimal} many-to-one matching, i.e., maximize the geometric mean of workers' and firms' utilities.
\begin{figure}[t]
    \centering
    \begin{subfigure}[b]{0.3\textwidth}
         \centering
         \begin{tikzpicture}
                \tikzset{firm/.style = {shape=rectangle,draw,inner sep=2.5pt}}
                \tikzset{worker/.style = {shape=circle,draw,inner sep=1.5pt}}
                \tikzset{edge/.style = {solid}}
                \tikzset{thickedge/.style = {solid, ultra thick}}
                %
                \node[firm] (1) at (0,0) {};
                \node (11) at (-0.5,0) {\color{black}{$f_2$}};
                \node[firm] (2) at (0,1) {};
                \node (22) at (-0.5,1) {\color{black}{$f_1$}};
                \node[worker] (3) at (2,0) {};
                \node (33) at (2.5,0) {\color{black}{$w_2$}};
                \node[worker] (4) at (2,1) {};
                \node (44) at (2.5,1) {\color{black}{$w_1$}};
                \draw[thickedge] (1) to node [near start,fill=white,inner sep=0pt] (123) {\scriptsize{$3$}} (3);
                \draw[edge] (1) to node [near start,fill=white,inner sep=0pt] (124) {\scriptsize{$2$}} (4);
                \draw[edge] (2) to node [near start,fill=white,inner sep=0pt] (223) {\scriptsize{$2$}} (3);
                \draw[thickedge] (2) to node [near start,fill=white,inner sep=0pt] (224) {\scriptsize{$3$}} (4);
    \end{tikzpicture}
     \end{subfigure}
     \begin{subfigure}[b]{0.3\textwidth}
         \centering
         \begin{tikzpicture}
                \tikzset{firm/.style = {shape=rectangle,draw,inner sep=2.5pt}}
                \tikzset{worker/.style = {shape=circle,draw,inner sep=1.5pt}}
                \tikzset{edge/.style = {solid}}
                \tikzset{thickedge/.style = {solid, ultra thick}}
                %
                \node[firm] (1) at (0,0) {};
                \node (11) at (-0.5,0) {\color{black}{$f_2$}};
                \node[firm] (2) at (0,1) {};
                \node (22) at (-0.5,1) {\color{black}{$f_1$}};
                \node[worker] (3) at (2,0) {};
                \node (33) at (2.5,0) {\color{black}{$w_2$}};
                \node[worker] (4) at (2,1) {};
                \node (44) at (2.5,1) {\color{black}{$w_1$}};
                \draw[edge] (1) to node [near start,fill=white,inner sep=0pt] (123) {\scriptsize{$3$}} node [near end,fill=white,inner sep=0pt] (321) {\scriptsize{$0$}} (3);
                \draw[thickedge] (1) to node [near start,fill=white,inner sep=0pt] (124) {\scriptsize{$2$}} node [near end,fill=white,inner sep=0pt] (421) {\scriptsize{$2$}} (4);
                \draw[thickedge] (2) to node [near start,fill=white,inner sep=0pt] (223) {\scriptsize{$2$}} node [near end,fill=white,inner sep=0pt] (223) {\scriptsize{$2$}} (3);
                \draw[edge] (2) to node [near start,fill=white,inner sep=0pt] (224) {\scriptsize{$3$}} node [near end,fill=white,inner sep=0pt] (424) {\scriptsize{$0$}} (4);
    \end{tikzpicture}
     \end{subfigure}
    \caption{One-sided (left) and two-sided (right) instances. For each edge, the number close to a vertex denotes how much that agent values the other agent. The Nash welfare maximizing matchings are highlighted via thick edges.}
    \label{fig:Motivating_Example}
\end{figure}
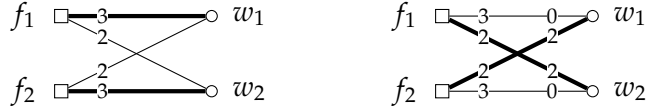

Before discussing our results, let us illustrate a key difference between one-sided and two-sided preferences through the example in \Cref{fig:Motivating_Example}. There are two workers $w_1,w_2$ and two firms $f_1,f_2$, each with capacity 1. The assignment $\{(f_1,w_1),(f_2,w_2)\}$ maximizes Nash welfare in the one-sided instance (ref. left subfigure). But when the workers' preferences are taken into account (ref. right subfigure), this assignment turns out to be \emph{arbitrarily} suboptimal, giving zero values to all workers and having zero Nash welfare. Thus, the structure of Nash optimal solutions for two-sided instances can drastically differ from their one-sided counterparts, which makes the two-sided problem challenging.

\Cref{tab:Results} summarizes our results. Some of the technical highlights of our work are discussed below.

\begin{table}[t]
\centering
\scriptsize
\begin{tabular}{|c|c|c|c|c|c|c|}
 \hline
 \textbf{\# firms} & \textbf{Capacities} & \textbf{Valuations} & \textbf{Exact/Approx.} & \textbf{Complexity} & \textbf{Technique} & \textbf{Ref.}\\
 \hline
 \rowcolor{myred}
 --- & 2 & Ternary $\{0,1,2\}$ & Exact & \NPH{} & Rainbow Matching & Thm.~\ref{thm:Nash-NPhard-Capacity-Two}\\
 \hline
 \rowcolor{myred}
 2 & Equal & Identical, symmetric & Exact & \NPH{} & Partition & Prop.~\ref{prop:Nash-NPhard-Symmetric-Vals}\\
 \hline
 %
 \rowcolor{myblue}
 --- & --- & Positive & $\nicefrac{1}{\sqrt{\opt}}$ & \Poly{} & Submodularity & Thm.~\ref{thm:Nash-Approx-OPT-Dependent}\\
 \hline
 \rowcolor{myblue}
 --- & Constant & Positive, constant & Constant approx. & \Poly{} & Submodularity & Cor.~\ref{cor:Nash-Approx-Constant-Vals}\\
 \hline
 \rowcolor{myblue}
 Constant & --- & $\OO(\textup{poly}(m,n))$ & $\nicefrac{1}{(1+\eps)}$-approx. & Quasipoly time & Bucketing & Thm.~\ref{thm:QPTAS-Constant-Firms}\\
 \hline
 \rowcolor{mygreen}
 Constant & --- & Constant & Exact & \Poly{} & Bucketing & Cor.~\ref{cor:FPT-Exact-Constant-Firms}\\
 \hline
 \rowcolor{mygreen}
 --- & --- & --- & Exact & \fpt{}, $\OO^\star(3^m)$ & Dynamic prog. & Thm.~\ref{thm:3m}\\
 \hline
 \rowcolor{mygreen}
 --- & Constant & --- & Exact & \fpt{}, $\OO^\star(2^m)$ & Dynamic prog. & Thm.~\ref{thm:param-m-const-capacity}\\
 \hline
 \rowcolor{myblue}
 --- & --- & --- & $\scriptstyle (1+\eps)^{\frac{-(n+1)}{(m+n)}}$-approx. & \fpt{}, $\OO^\star(\nicefrac{1}{\eps^4} \cdot 2^m)$ & Poly. multiplication & Thm.~\ref{thm:param-m-fptas}\\
 \hline
 \rowcolor{mygreen}
 --- & --- & Symmetric binary & Exact & \Poly{} & Greedy algorithm & Thm.~\ref{thm:symm-Binary-Vals}\\
 \hline
 \rowcolor{mygreen}
 --- & --- & --- & Exact & \Poly{} for deg $\leq 2$ & Paths and cycles & Thm.~\ref{thm:Combined-Restricted-Domains}\\
 \hline
 \rowcolor{mygreen}
 --- & Exactly 2 & --- & Exact & \Poly{} for deg $\leq 3$ & Reduction rules & Thm.~\ref{thm:Combined-Restricted-Domains}\\
 \hline
\end{tabular}
\vspace{0.1in}
\caption{Summary of our results on maximizing Nash welfare with $n$ firms and $m$ workers. Each row lists a result under the assumptions specified by the individual columns. A ``---'' means that the parameter for that column can be arbitrary. The red colored rows contain hardness results, blue rows contain approximation algorithms, and green rows contain exact algorithms.}
\label{tab:Results}
\end{table}

\paragraph{Hardness Results.} In \Cref{sec:Hardness_Results}, we show that finding a Nash optimal matching is \NPH{} even when each firm has capacity $2$ and the valuations are \emph{ternary} in the range~$\{0,1,2\}$~(\Cref{thm:Nash-NPhard-Capacity-Two}). This hardness result is ``tight'' in the sense that further restricting either of these assumptions---firms having unit capacities or valuations being \emph{symmetric} and \emph{binary} $\{0,1\}$---allows efficient algorithms~(\Cref{prop:One-One,thm:symm-Binary-Vals}). 
%
%

\paragraph{Approximation Algorithms.} In \Cref{sec:Approximation_Algorithms}, we provide two approximation algorithms. The first one is a $\frac{1}{\sqrt{\opt}}$-approximation algorithm, where $\opt$ is the optimal Nash welfare~(\Cref{thm:Nash-Approx-OPT-Dependent}). This result assumes that the valuations are positive integers, and uses a reduction to the fair division problem under \emph{submodular} valuations and \emph{matroid} constraints. 
%
%
The second algorithm uses a discretization technique to give a quasipolynomial-time approximation scheme (\QPTAS{}) for a constant number of firms and polynomially-bounded valuations~(\Cref{thm:QPTAS-Constant-Firms}). Note that the problem is \NPH{} even for 
\emph{two} firms~(\Cref{prop:Nash-NPhard-Symmetric-Vals}). Whether this hardness implication can be achieved for polynomially-bounded valuations remains unresolved.
    
\paragraph{Parameterized Algorithms.} In \Cref{sec:Parameterized Algorithms}, we provide two parameterized algorithms. The first algorithm, based on dynamic programming, is fixed-parameter tractable~(\fpt) in the number of workers $m$ and has a running time of $\OO^\star(3^m)$~(\Cref{thm:3m}); here, the notation $\OO^\star(\cdot)$ suppresses multiplicative polynomial terms. 
The second algorithm is an \fpt{} approximation scheme with a faster running time of $\OO^\star(2^m)$~(\Cref{thm:param-m-fptas}). This algorithm uses polynomial multiplication; to our knowledge, this is the first use of this technique in the context of Nash welfare.

\paragraph{Restricted Domains.} Finally, in \Cref{sec:Restricted_Domains}, we develop exact, polynomial-time algorithms for restricted settings such as symmetric binary valuations~(\Cref{thm:symm-Binary-Vals}) and bounded degree instances~(\Cref{thm:Combined-Restricted-Domains}). Our algorithm for symmetric binary valuations uses the greedy algorithm of \citet{BKV18greedy} from fair division for maximizing Nash welfare under binary valuations.

\subsection{Related Work}
\label{subsec:Related_Work}

We will now briefly survey the relevant literature on Nash welfare in fair division. A detailed discussion of the related work can be found in the appendix.

For fair division of \emph{divisible} resources, the well-known Eisenberg-Gale convex program is known to efficiently compute a fractional allocation that maximizes Nash welfare~\citep{EG59consensus}. By contrast, for \emph{indivisible} resources, the problem becomes \APXH{} under additive valuations~\citep{NNR+14computational,L17apx}. Note that the fair division problem for indivisible resources is a special case of our model when every worker (i.e., the `items') values all firms (i.e., the `agents') at $1$, and the firms have \emph{unrestricted} capacities. For the case of restricted capacities, which is the focus of our study, the hardness constructions from fair division do not automatically carry over to our setting~(we discuss some exceptions to this remark in \Cref{sec:Hardness_Results}). 

On the algorithmic front, for additive valuations, a fully polynomial-time approximation scheme (FPTAS) is known for any fixed number of agents~\citep{NNR+14computational,GHM22tractable}, and constant-factor approximation algorithms are known for a general number of agents~\citep{CG18approximating,CDG+17convex,BKV18finding}. Several works have studied approximation algorithms for more general classes of valuations such as \emph{budget-additive}~\citep{GHM18approximating}, \emph{separable piecewise-linear concave}~\citep{AMG+18nash}, \emph{submodular}~\citep{GKK20approximating,GHL+22approximating}, \emph{XOS}~\citep{BKK+21sublinear}, 
and \emph{subadditive}~\citep{BBK+20tight,CGM21fair}. For more details on fair division with indivisible resources, we refer the reader to the survey by~\citet{AAB+23fair}.

\section{Preliminaries}
\label{sec:Preliminaries}

For any positive integer $r$, let $[r] \coloneqq \{1,2,\dots,r\}$.

\paragraph{Problem instance.} An instance of the \emph{two-sided matching problem} is given by a tuple $\langle W,F,C, \V \rangle$, where $W = \{w_1,\dots,w_m\}$ is the set of $m$ \emph{workers}, $F = \{f_1,\dots,f_n\}$ is the set of $n$ \emph{firms}, $C = \{c_1,\dots,c_n\}$ is the set of \emph{capacities} of the firms, and $\V = (v_{w_1},\dots,v_{w_m},v_{f_1},\dots,v_{f_n})$ is the \emph{valuation profile} consisting of the valuation function of each worker and firm. For each worker $w \in W$, its valuation function $v_w : F \cup \{\emptyset\} \rightarrow \N \cup \{0\}$ specifies its numerical value for each firm, where $ v_w(\emptyset) \coloneqq 0$. For each firm $f \in F$, its valuation function $v_f : 2^W \rightarrow \N \cup \{0\}$ specifies its numerical value for each subset of workers, where $ v_f(\emptyset) \coloneqq 0$.

\paragraph{Many-to-one matching.} A \emph{many-to-one} matching $\mu: W \times F \rightarrow \{0,1\}$ is a function that assigns each worker-firm pair a weight of either $0$ or $1$ such that:
\begin{itemize}
    \item for every worker $w \in W$, $\sum_{f \in F} \mu(w,f) \leq 1$, i.e., each worker is matched with at most one firm, and
    \item for every firm $f \in F$, $\sum_{w \in W} \mu(w,f) \leq c_f$, i.e., each firm $f$ is matched with at most $c_f$ workers.
\end{itemize}

We will write $\mu(w)$ to denote the firm that worker $w$ is matched with under $\mu$, i.e, $\mu(w) \coloneqq \{f \in F : \mu(w,f)=1\}$. Similarly, we will write $\mu(f)$ to denote the set of workers that are matched with firm $f$ under $\mu$, i.e., $\mu(f) \coloneqq \{w \in W : \mu(w,f) = 1\}$. Thus, 
$|\mu(w)| \leq 1$ and $|\mu(f)| \leq c_f$.

For the special case when $c_f = 1$ for all firms $f \in F$, we obtain the \emph{one-to-one} matching problem.

For simplicity, we will use the term \emph{matching} to denote a `many-to-one matching', and will explicitly use the qualifiers `one-to-one' and `many-to-one' when the distinction between the two is necessary. Further, the term \emph{agent} will refer to a worker or a firm, i.e., an entity in the set $W \cup F$.

\paragraph{Utilities.} Given a matching $\mu$, the \emph{utility} of a worker $w$ under $\mu$ is its value for the firm that it is matched with, i.e., $u_w(\mu) \coloneqq \sum_{f \in F} v_{w,f} \cdot \mu(w,f)$. Similarly, the utility of a firm $f$ under $\mu$ is the sum of its values for the workers that it is matched with, i.e., $u_f(\mu) \coloneqq \sum_{w \in W} v_{f,w} \cdot \mu(w,f)$. In other words, the firms' utilities are assumed to be \emph{additive}.

\paragraph{Welfare measures.} We will now define two welfare measures that can be associated with a matching $\mu$. 
\begin{itemize}
    \item \emph{Utilitarian social welfare} is the sum of the utilities of the agents under $\mu$, i.e., $\W^\texttt{\textup{util}}(\mu) \coloneqq \sum_{i \in W \cup F} u_i(\mu)$.
    \item \emph{Nash social welfare} is the geometric mean of the utilities of the agents under $\mu$, i.e., $$\W^\texttt{\textup{Nash}}(\mu) \coloneqq \left( \prod_{i \in W \cup F} \, u_i(\mu) \right)^{\nicefrac{1}{n+m}}.$$ We will use the term \emph{Nash product} to denote the product of agents' utilities.
\end{itemize}

For any $\alpha \in [0,1]$ and any welfare measure $\W$, a matching $\mu$ is said to be \emph{$\alpha$-$\W$ optimal} if its welfare is at least $\alpha$ times the maximum welfare achieved by any matching for the given instance. When $\alpha=1$, we use the term \emph{$\W$ optimal}, e.g., utilitarian optimal or Nash optimal. 

\paragraph{The case of zero Nash welfare.} Since we allow the agents to have zero valuations, it is possible that every matching for a given instance has zero Nash welfare. In the appendix, we give a polynomial-time algorithm for identifying such instances. This algorithm uses a natural linear program for many-to-one matchings along with the rounding technique of~\citet{BCK+13designing} and~\citet{AFS+23best}.

\begin{restatable}{proposition}{ZeroNash}
There is a polynomial-time algorithm that, given any two-sided matching instance, determines whether there exists a matching with nonzero Nash welfare.
\label{prop:Zero_Nash}
\end{restatable}

The `zeroness' of Nash welfare turns out to be important in analyzing its axiomatic properties in the fair division literature. Indeed, for an instance where all allocations have zero Nash welfare, a Nash optimal allocation is defined in terms of the largest set of agents that can have positive utility~\citep{CKM+19unreasonable}. This refinement is necessary for demonstrating the approximate envy-freeness property of a Nash optimal allocation. Since our focus in this work is on computational---rather than axiomatic---properties, we will assume that instances where the optimal Nash welfare is zero are discarded. Thus, in the forthcoming sections, we will assume that any given instance admits some matching with nonzero Nash welfare. Equivalently, there is some matching in which all agents have positive utility. Note that this condition enforces that the number of workers is at most the total capacity of the firms, i.e., $m \leq \sum_{f \in F} c_f$.

\paragraph{Fair division as a special case.}
An important special case of our two-sided model is the problem of fair division with indivisible items. An instance of this problem consists of a set of agents with numerical valuations over a set of indivisible items. The goal is to partition the items among the agents subject to fairness constraints, such as finding an allocation that maximizes Nash welfare. When each worker values every firm at $1$ and the firms have unrestricted capacities, we obtain fair division as a special case of our problem.

\paragraph{Relevant parameters.}

In addition to natural parameters such as the number of workers, number of firms, and firms' capacities, we will use other structural parameters in our analysis. Consider the bipartite graph (similar to the one shown in \Cref{fig:Motivating_Example}) associated with any matching instance. Suppose all ``0---0'' edges, wherein a worker and firm value each other at 0, are removed. The \emph{degree} of an agent 
is defined as the number of edges incident to the corresponding vertex in the residual graph. Thus, the degree of an agent gives an upper bound on the number of other agents that it positively values. We also consider the \emph{number of distinct valuations} parameter, defined as the number of distinct elements in the set~$\{v_{w,f}\}_{ (w,f) \in W \times F} \cup \{v_{f,w}\}_{ (w,f) \in W \times F}$.

All omitted proofs are presented in the appendix.

\section{Hardness Results}
\label{sec:Hardness_Results}

In this section, we will prove our main negative result (\Cref{thm:Nash-NPhard-Capacity-Two}) concerning the hardness of maximizing Nash welfare under two-sided preferences.

To motivate our hardness result, let us first consider an easy case where the capacity of each firm is $1$. In this case, a Nash optimal matching can be computed in polynomial time~(\Cref{prop:One-One}). The algorithm is fairly natural: It computes a maximum weight matching in a bipartite graph whose vertices are the workers and the firms, and the weight of the edge between worker $w$ and firm $f$ is $\log (v_{w,f} \cdot v_{f,w})$ if both $v_{w,f}$ and $v_{f,w}$ are positive, and $0$ otherwise.

\begin{restatable}{proposition}{}
When every firm has capacity $1$, a Nash optimal matching can be computed in polynomial time.
\label{prop:One-One}
\end{restatable}

A similar observation has been previously made in the fair division problem with one-sided preferences~\citep{GKK20approximating}. \Cref{prop:One-One} extends this idea to the two-sided setting.

Interestingly, the problem becomes significantly harder when the firms can have capacity $2$ (\Cref{thm:Nash-NPhard-Capacity-Two}). On first thought, it may seem that the hardness arises because of \emph{large} valuations. Or, if the valuations are small, the difficulty may be due to balancing a \emph{large number} of small values for each firm. However, we show that the hardness persists even when all valuations are in the range $\{0,1,2\}$ (also known as \emph{ternary} or \emph{3-value} instances) and each agent has \emph{degree} at most $3$; thus, it positively values at most three other agents.

\begin{restatable}[\textbf{NP-hardness under constant capacities}]{theorem}{NPHardCapacityTwo}
Computing a Nash optimal matching is \NPC{} even when each firm has capacity $2$, all valuations are in $\{0,1,2\}$, and each agent has degree at most $3$.
\label{thm:Nash-NPhard-Capacity-Two}
\end{restatable}

Before discussing the proof of \Cref{thm:Nash-NPhard-Capacity-Two}, let us compare it with some known results on the hardness of maximizing Nash welfare in the fair division literature. 
%
\citet{ABF+21maximum} have shown that 
maximizing Nash welfare is \NPH{} for 3-value instances when all valuations are in the range $\{0,1,a\}$ for some $a>1$. The parameter `$a$' in their construction \emph{depends} on the size of the instance; specifically, they use $a > 1/ \sqrt[2m]{2} -1$, where $m$ is the number of clauses in the \ThreeSAT{} instance which they reduce from. By contrast, our reduction \emph{does not} require the values to grow with the size of the instance; indeed, all valuations in our construction are in the range~$\{0,1,2\}$. 


\citet{{NNR+14computational}} showed \NPH{}ness of approximating Nash welfare to within a factor $\frac{8}{9}$ via reduction from \ExactCoverByThreeSets{}. Their reduction requires the capacities to be $3$ (instead of $2$) and the number of positively valued agents to grow with the problem size (instead of being constant). Finally,~\citet{GHM18approximating} showed \NPH{}ness of approximating Nash welfare to within a factor $\sqrt{7/8}$ via a reduction from a variant of \MaxLinEq{} problem. The capacities in their reduction are required to be $4$ (instead of $2$) and the valuations are in $\{0,1,k\}$ for a large constant $k$ (even under a conservative calculation, $k \geq 8$).

Let us now briefly sketch the proof of~\Cref{thm:Nash-NPhard-Capacity-Two}. 
\begin{proofsketch} (of \Cref{thm:Nash-NPhard-Capacity-Two}) Our reduction is from \RainbowPerfectMatching{}. The input to this problem consists of a bipartite multigraph whose vertex sets are of size $r$ each and whose edge set is partitioned into~$r$ color classes. The goal is to determine if there is a perfect matching consisting of one edge of each color. In the appendix, we show that this problem is \NPH{} even when each vertex has degree $3$ and there are three edges of each color. These restrictions on \RainbowPerfectMatching{} are needed in our proof for obtaining the bound on degree in~\Cref{thm:Nash-NPhard-Capacity-Two}.

The reduced instance consists of a \emph{main firm} for every edge and a \emph{main worker} for every vertex of the given multigraph. Additionally, there are three \emph{dummy workers} and one \emph{dummy firm} for each color class. The main workers only value the main firm corresponding to their incident edges, while the main firms have a slightly higher value for dummy workers; see \Cref{fig:Nash-NPhard-Capacity-Two}. Each dummy firm only values the dummy workers of the same color class and has zero value for all other workers. The capacity of every firm is $2$, and the threshold for Nash welfare is $\theta \coloneqq 2^{\# \textup{firms}/\# \textup{agents}}$. That is, the goal is to determine if there is a feasible matching with Nash welfare at least $\theta$. Note that there are $4r$ firms and $5r$ workers in the reduced instance; thus $9r$ agents in total.

Given a rainbow perfect matching, the desired Nash welfare solution can be obtained by assigning the main workers to the firms corresponding to the matched edges. This leaves two main firms in each color class, who each absorb one dummy worker, leaving the third dummy worker for a dummy firm. Thus, all firms get a utility of $2$, as desired.

For the reverse direction, we use the AM-GM inequality. Indeed, each dummy worker contributes $2$ to the sum of firms' utilities, while each main worker contributes $1$. Thus, the arithmetic mean of firms' utilities is \emph{exactly} $2$. The decision threshold of $\theta$ forces the geometric mean to also be exactly $2$, implying that every firm gets a utility of $2$. This naturally induces a rainbow matching.
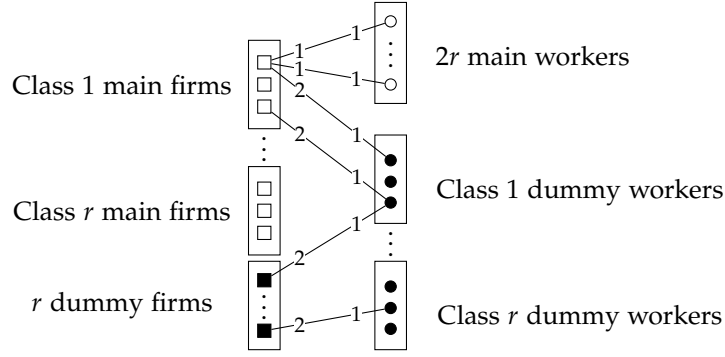
\begin{figure}[t]
    \centering
    \small
    \begin{tikzpicture}[scale=0.84]
                \tikzset{firm/.style = {shape=rectangle,draw,inner sep=2.5pt}}
                \tikzset{dummyfirm/.style = {shape=rectangle,fill=black,draw,inner sep=2.5pt}}
                \tikzset{worker/.style = {shape=circle,draw,inner sep=1.5pt}}
                \tikzset{dummyworker/.style = {shape=circle,draw,fill=black,inner sep=1.5pt}}
                \tikzset{edge/.style = {solid}}
                \draw[draw=black] (0,3.5) rectangle ++(.5,1.4);
                \node (1) at (-2,4.2) {\color{black}{Class 1 main firms}};
                \node[] (0) at (0.25,3.3) {$\vdots$};
                \draw[draw=black] (0,1.5) rectangle ++(.5,1.4);
                \node (1) at (-2,2.2) {\color{black}{Class $r$ main firms}};
                \draw[draw=black] (0,0) rectangle ++(.5,1.4);
                \node (1) at (-2,0.7) {\color{black}{$r$ dummy firms}};
                \draw[draw=black] (2,0) rectangle ++(.5,1.4);
                \node (1) at (5.25,0.5) {\color{black}{Class $r$ dummy workers}};
                \node[] (0) at (2.25,1.8) {$\vdots$};
                \draw[draw=black] (2,2) rectangle ++(.5,1.4);
                \node (1) at (5.25,2.5) {\color{black}{Class 1 dummy workers}};
                \draw[draw=black] (2,3.9) rectangle ++(.5,1.6);
                \node (1) at (4.5,4.65) {\color{black}{$2r$ main workers}};
                \node[dummyfirm] (1) at (0.25,0.3) {};
                \node[] (0) at (0.25,0.8) {$\vdots$};
                \node[dummyfirm] (3) at (0.25,1.1) {};
                \node[firm] (4) at (0.25,1.85) {};
                \node[firm] (5) at (0.25,2.2) {};
                \node[firm] (6) at (0.25,2.55) {};
                \node[firm] (7) at (0.25,3.85) {};
                \node[firm] (8) at (0.25,4.2) {};
                \node[firm] (9) at (0.25,4.55) {};
                \node[dummyworker] (10) at (2.25,0.33) {};
                \node[dummyworker] (11) at (2.25,0.66) {};
                \node[dummyworker] (12) at (2.25,1) {};
                \node[dummyworker] (13) at (2.25,2.33) {};
                \node[dummyworker] (14) at (2.25,2.66) {};
                \node[dummyworker] (15) at (2.25,3) {};
                \node[worker] (16) at (2.25,4.2) {};
                \node[] (0) at (2.25,4.8) {$\vdots$};
                \node[worker] (18) at (2.25,5.2) {};
                \draw[edge] (9) to node [near start,fill=white,inner sep=0pt] (9218) {\scriptsize{$1$}} node [near end,fill=white,inner sep=0pt] (1829) {\scriptsize{$1$}} (18);
                \draw[edge] (9) to node [near start,fill=white,inner sep=0pt] (9216) {\scriptsize{$1$}} node [near end,fill=white,inner sep=0pt] (1629) {\scriptsize{$1$}} (16);
                \draw[edge] (9) to node [near start,fill=white,inner sep=0pt] (8215) {\scriptsize{$2$}} node [near end,fill=white,inner sep=0pt] (1528) {\scriptsize{$1$}} (15);
                \draw[edge] (1) to node [near start,fill=white,inner sep=0pt] (1211) {\scriptsize{$2$}} node [near end,fill=white,inner sep=0pt] (1121) {\scriptsize{$1$}} (11);
                \draw[edge] (3) to node [near start,fill=white,inner sep=0pt] (3213) {\scriptsize{$2$}} node [near end,fill=white,inner sep=0pt] (1323) {\scriptsize{$1$}} (13);
                \draw[edge] (7) to node [near start,fill=white,inner sep=0pt] (7213) {\scriptsize{$2$}} node [near end,fill=white,inner sep=0pt] (1327) {\scriptsize{$1$}} (13);
    \end{tikzpicture}
    \caption{The reduced two-sided matching instance in the proof of \Cref{thm:Nash-NPhard-Capacity-Two}. Firms are denoted by squares and workers by circles. The shaded and unshaded nodes denote the dummy and the main agents, respectively.}
    \label{fig:Nash-NPhard-Capacity-Two}
\end{figure}
\end{proofsketch}

By reducing from the \Partition{} problem, we can show \NPH{}ness of maximizing Nash welfare even for two firms with identical valuations and equal capacities and even under \emph{symmetric} valuations (i.e., if for every worker $w$ and firm $f$, $v_{w,f} = v_{f,w}$). This reduction is well-known in the fair division literature; \Cref{prop:Nash-NPhard-Symmetric-Vals} simply adapts it to two-sided preferences setting.

\begin{restatable}[\textbf{NP-hardness for two identical firms}]{proposition}{SymmetricValsNPhard}
Computing a Nash optimal matching is \NPC{} even with two identical firms (that have the same valuation functions and equal capacities) and under symmetric valuations.
\label{prop:Nash-NPhard-Symmetric-Vals}
\end{restatable}

A slight modification of the same reduction gives \NPH{}ness for the case when all agents have \emph{strict} preferences.

\begin{restatable}[\textbf{NP-hardness under strict preferences}]{corollary}{StrictPrefsNPhard}
Computing a Nash optimal matching is \NPC{} even when all agents have strict preferences.
\label{cor:Nash-NPhard-Strict-Vals}
\end{restatable}

\section{Approximation Algorithms}
\label{sec:Approximation_Algorithms}

The intractability of maximizing Nash welfare
~(\Cref{sec:Hardness_Results}) motivates the study of approximation algorithms in search of positive results. To this end, we provide two algorithms: 
The first algorithm works for an arbitrary number of firms, but gives a somewhat weak approximation~(\Cref{thm:Nash-Approx-OPT-Dependent}). The second algorithm provides a stronger approximation for a constant number of firms~(\Cref{thm:QPTAS-Constant-Firms,cor:FPT-Exact-Constant-Firms}).

Our first main result in this section is a $\frac{1}{\sqrt{\opt}}$-approximation algorithm, where $\opt$ is the optimal Nash welfare in the given instance. The algorithm requires all valuations to be positive, i.e., for every worker $w \in W$ and every firm $f \in F$, $v_{w,f}>0$ and $v_{f,w} > 0$.

\begin{restatable}[\textbf{$\frac{1}{\sqrt{\opt}}$-approximation}]{theorem}{NashApproxOPTdependent}
There is a polynomial-time algorithm that, given any matching instance with positive valuations, 
returns a $\frac{1}{\sqrt{\opt}}$-Nash optimal matching, where $\opt$ is the optimal Nash welfare. 
\label{thm:Nash-Approx-OPT-Dependent}
\end{restatable}

The approximation factor in \Cref{thm:Nash-Approx-OPT-Dependent} \emph{decreases} as a function of the optimal Nash welfare and provides only a weak guarantee as the problem size grows. Nevertheless, when the valuations are bounded by a constant and each firm has a constant capacity, the maximum utility of any agent---and thus their geometric mean---is also constant. In this setting, \Cref{thm:Nash-Approx-OPT-Dependent} gives a \emph{constant-factor} approximation.

\begin{restatable}[\textbf{Constant approximation for constant valuations and capacities}]{corollary}{}
Let $c_0$ be a constant. There is a polynomial-time algorithm that, given any matching instance with positive valuations where all valuations and capacities are bounded by $c_0$, returns an $\alpha$-Nash optimal matching for some constant $\alpha$ (that only depends on $c_0$).
\label{cor:Nash-Approx-Constant-Vals}
\end{restatable}

Recall that even when all valuations and capacities are bounded by a constant, the problem of maximizing Nash welfare is \NPH{}~(\Cref{thm:Nash-NPhard-Capacity-Two}). Whether there exists a polynomial-time approximation scheme (\PTAS{}) in this setting is an interesting open problem. Later in this section, we will present a \emph{quasi}polynomial-time approximation scheme or \QPTAS{} for a constant number of firms~(\Cref{thm:QPTAS-Constant-Firms}).

\begin{proofsketch} (of \Cref{thm:Nash-Approx-OPT-Dependent})
We will show that the task of approximating Nash welfare in the matching problem reduces to approximating \emph{utilitarian} welfare in the \emph{fair division} problem under \emph{submodular} valuations and \emph{matroid} constraints. Specifically, we let the firms and workers play the role of agents and items, respectively, in the fair division problem. For each firm $f$ and any subset $X \subseteq W$ of workers, the \emph{modified} valuation of `agent' $f$ for the set of `items' $X$ is defined as $\hat{v}_f(X) \coloneqq \log \left( v_f(X) \cdot \Pi_{w \in X} v_{w,f} \right)$.

The modified valuation function $\hat{v}_f(\cdot)$ captures the combined contribution of firm $f$ and the set of workers $X$ to the log of Nash welfare objective. It is easy to show that an $\alpha$-utilitarian optimal allocation 
(under modified valuations) induces a matching that is $\frac{1}{\opt^{1-\alpha}}$-Nash optimal.

The key observation is that the function $\hat{v}_f(\cdot)$ is \emph{submodular}. Furthermore, capacity constraints in the matching problem can be captured by matroid constraints in the fair division problem.
We can now use the natural greedy algorithm for submodular maximization 
to obtain a feasible and $\frac{1}{2}$-utilitarian optimal allocation with respect to the modified valuations~\citep{FNW78analysis}. The desired approximation for Nash welfare follows.
\end{proofsketch}

It would be very interesting to know if a constant-factor approximation algorithm can be designed for our problem. 
Such algorithms are known for the one-sided fair division problem. However, as the example in \Cref{fig:Motivating_Example} suggests, incorporating the preferences of the `items' can come at the expense of making the `agents' worse off. This makes the design of algorithms for the two-sided problem challenging.

Our second main result in this section uses the idea of \emph{bucketing} (or discretization) to classify the workers and firms according to the \emph{range} in which they value each other. 

Specifically, given any $\eps > 0$ and any firm $f$, define a set of $\tau+1$ \emph{buckets} $B_0^f,B_1^f,B_2^f,\dots,B_{\tau}^f$, where $\tau \coloneqq \lceil \log_{1+\eps} v_{\max} \rceil$ and $v_{\max}$ is the maximum valuation of any agent for any other agent. For any $i \in [\tau]$, the bucket $B_i^f$ denotes the set of workers who value firm $f$ between $(1+\eps)^{i-1}$ and $(1+\eps)^{i}$. Bucket $B_0^f$ contains workers who value firm $f$ at $0$. For each bucket $B_i^f$, we further define $\tau+1$ \emph{sub-buckets} $b_0^{f,i},b_1^{f,i},\dots,b_\tau^{f,i}$, where, for any $j \in [\tau]$, the sub-bucket $b_j^{f,i}$ denotes the set of workers in the bucket $B_i^f$ whom the firm $f$ values between $(1+\eps)^{j-1}$ and $(1+\eps)^{j}$. The workers in $B_i^f$  valued at $0$ are assigned to $b_0^{f,i}$. Thus, membership in a sub-bucket specifies the valuations of a worker and a firm for each other within a multiplicative factor of $(1+\eps)$.

Our algorithm \emph{guesses} the number of workers in each sub-bucket in an optimal solution. Each such guess, if feasible for the given capacities, induces a matching. The Nash welfare of this matching can be correctly computed to within a multiplicative factor of $(1+\eps)$ by only knowing the \emph{number} of workers in each sub-bucket. 
The total number of sub-buckets 
is $\OO(n \tau^2)$. For $m$ workers, the total number of guesses is at most $m^{\OO(n \tau^2)}$; the calculation is analogous to distributing $m$ identical 
balls into $n (\tau+1)^2$ bins. For each guess, the algorithm checks feasibility and computes Nash welfare in polynomial time. The guess with the maximum Nash welfare is returned as the solution. For \emph{polynomially-bounded} valuations, i.e., when $v_{\max} \leq \textup{poly}(m,n)$, we obtain a \QPTAS{} for a constant number of firms.

\begin{restatable}[\textbf{QPTAS for constant no. of firms}]{theorem}{QPTASalgorithm}
There is an algorithm that, given any $\eps > 0$ and any matching instance with a constant number of firms and polynomially-bounded valuations, runs in $\OO(m^{\OO({\frac{1}{\eps^2}\cdot \log^2 m})})$ time and returns a $\frac{1}{(1+\eps)}$-Nash optimal matching.
\label{thm:QPTAS-Constant-Firms}
\end{restatable}

By bucketing on the \emph{exact} valuations instead of powers of $(1+\eps)$, we obtain a polynomial-time exact algorithm for the case when $n$ and $v_{\max}$ are both constant.

\begin{restatable}[\textbf{Exact algorithm for constant no. of firms and constant valuations}]{corollary}{ExactBucketing}
There is an algorithm that, given any matching instance with a constant number of firms and the maximum valuation $v_{\max}$ bounded by a constant, runs in $\OO(\textup{poly}(m))$ time and returns a Nash optimal matching.
\label{cor:FPT-Exact-Constant-Firms}
\end{restatable}

Note that the problem is \NPH{} for a constant number of firms but with possibly large valuations~(\Cref{prop:Nash-NPhard-Symmetric-Vals}).

\section{Parameterized Algorithms}
\label{sec:Parameterized Algorithms}

We will now study the problem in the realm of parameterized complexity and discuss two main results: An exact algorithm that is fixed-parameter tractable (\fpt{}) in the number of workers~$m$ with running time $\OO^\star(3^m)$~(\Cref{thm:3m}), and a faster \fpt{} approximation algorithm with running time $\OO^\star(2^m)$~(\Cref{thm:param-m-const-capacity}).

We have already shown in \Cref{sec:Hardness_Results} that maximizing Nash welfare is \NPH{} even when the maximum capacity of any firm~(denoted by $\capacity$) or the number of firms $n$ is constant. Thus, a natural next step is to study the combined parameter $n + \capacity$. By the nonzeroness of Nash welfare, we know that $m \leq n \cdot \capacity$. Thus, designing an \fpt{} algorithm in $m$ automatically gives an \fpt{} algorithm in $n \cdot \capacity$, and thus also \fpt{} in $n + \capacity$.

Note that designing \emph{some} \fpt{}-in-$m$ algorithm is straightforward: Since the number of firms is at most the number of workers (i.e., $n \leq m$), a brute force search over the space of all feasible matchings can be done in $\OO^\star(m^m)$ time. Our goal in this section is to obtain \fpt{}-in-$m$ algorithms with significantly better running times.







The following notation will be useful in proving both results in this section: For any firm $f$ and any subset of workers $S \subseteq W$, define ${\W}_{f}(S) \coloneqq v_f(S) \cdot \prod_{w\in S}v_w(f)$. Recall that the log of the function $\W_f(\cdot)$ was used in the proof of \Cref{thm:Nash-Approx-OPT-Dependent} in defining the modified valuation function.

Let us now state our first main result in this section.


\begin{restatable}[\textbf{\fpt{} exact algorithm}]{theorem}{DPalgorithm}
A Nash optimal matching can be computed in $\OO^\star(3^m)$ time. 
    \label{thm:3m}
\end{restatable}

To prove \Cref{thm:3m}, we will use dynamic programming. Let the firms be indexed as $f_1,\dots,f_n$. Define a table $T$ with $n$ rows (one for each firm) and $2^m$ columns (one for each subset of workers). The entry $T(i,S)$ will contain the value of the maximum Nash \emph{product} achievable in a subinstance consisting of the first $i$ firms and the subset $S$ of the workers.

More concretely, for every subset $S\subseteq W$, the first row of the table can be computed as follows:
\begin{equation*}
    \begin{split}
    T[1,S] = \begin{cases} \W_{f_1}(S) & \text{ if } |S|\leq c_1, \\
    0 & \text{ otherwise.}
    \end{cases}
    \end{split}
\end{equation*}
Further, for any $i>1$ and any subset $S\subseteq W$, we have
\begin{equation*}
    T[i,S] = \max_{S'\subseteq S, |S'|\leq c_i} \W_{f_i}(S') \times T[i-1,S\setminus S'].
\end{equation*}

Once the entry $T[n,W]$ is computed correctly, we can use backtracking to obtain the Nash optimal matching. To compute the entry $T[i,S]$, we may need to check all subsets of $S$ in the worst case. Thus, the running time of the algorithm is $\OO( n\sum_{i=0}^m \binom{m}{i}2^i )$, or $\OO( n \cdot 3^m )$.

The running time of our DP algorithm can be improved if the capacity of each firm is bounded by a constant. In this case, we only need to consider subsets of constant size, implying that each entry of the table can be computed in polynomial time. \Cref{thm:param-m-const-capacity} formalizes this observation.

\begin{restatable}[\textbf{$\OO^\star(2^m)$ algorithm for constant capacity}]{theorem}{}
When every firm has a constant capacity, a Nash optimal matching can be computed in $\OO^\star(2^m)$ time.
     \label{thm:param-m-const-capacity}
\end{restatable}
   

Next, we will present a parameterized approximation algorithm that, given any $\eps \in (0,1]$, finds a $(1+\eps)^{\nicefrac{-(n+1)}{m+n}}$-approximate solution in $\OO^\star(\nicefrac{1}{\eps^4} \cdot 2^m) $ time for arbitrary capacities. 
The algorithm uses the techniques of \emph{polynomial multiplication} and \emph{binning}. 

Given any instance~$\langle W,F,C,\V \rangle$, let $\mu$ and $\eta$ denote a Nash optimal matching and the maximum Nash product, respectively. 
For every $i\in [n]$, let ${\cal S}_i$ denote the set of all subsets of workers that can be feasibly assigned to the firm $f_i$, i.e., ${\cal S}_i$ contains all subsets of workers of size at most $c_i$. 

The basic idea is as follows: For each firm $f\in F$, guess $\W_f(\mu(f))$, remove the set $S$ from ${\cal S}_i$ if ${\W}_{f}(S)$ is not same as the guessed value, and then find $n$ disjoint sets, one from each ${\cal S}_i$. The disjoints sets can be found using polynomial multiplication in $\OO^\star(2^m)$ time. Unfortunately, this will not lead to the desired running time as guessing depends on the Nash product. Thus, we will create $(1+\eps)$-sized bins for the Nash product and obtain an approximation algorithm. 

Before we discuss our algorithm, let us introduce some relevant definitions. 
The {\em characteristic vector} of a subset $S\subseteq [m]$, denoted by $\chi(S)$, is an $m$-length binary string whose $i^{\textup{th}}$ bit is $1$ if and only if  $i \in S$.  Two binary strings of length $m$ are said to be \emph{disjoint} if for each $i\in [m]$, the $i^\textup{th}$ bits in the two strings are different. The {\em Hamming weight} of a binary string $S$, denoted by ${\cal H}(S)$, is the number of $1$s in the string $S$. A monomial $y^i$ is said to have Hamming weight $w$ if the degree term $i$, when represented as a binary string, has Hamming weight $w$. The {\em Hamming projection} of a polynomial $p(y)$ to $h$, denoted by ${\cal H}_{h}(p(y))$, is the sum of all the monomials of $p(y)$ which have Hamming weight $h$. 
The {\em representative polynomial} of $p(y)$, denoted by ${\cal R}(p(y))$, is a polynomial derived from $p(y)$ by changing the coefficients of all monomials contained in it to $1$. 
 
We begin with the following known result. 





\begin{restatable}[\citealp{DBLP:conf/ijcai/Gupta00T21,DBLP:journals/tcs/CyganP10}]{proposition}{HWDisjointness}
Subsets $S_{1}, S_{2} \subseteq W$ are disjoint if and only if Hamming weight of the monomial $x^{\chi(S_1)+\chi(S_2)}$ is $|S_{1}|+|S_{2}|$.
\label{cor:HW-disjointness}
\end{restatable}


Let us now state our main result.


\begin{restatable}[{\bf \fpt{} approximation scheme}]{theorem}{PolyMultThm}
   There is an algorithm that, given any matching instance and any $\eps \in (0,1]$, returns a $(1+\eps)^{\nicefrac{-(n+1)}{m+n}}$-Nash optimal matching in $\OO^\star(\nicefrac{1}{\eps^4} \cdot 2^m)$ time. 
   \label{thm:param-m-fptas}
\end{restatable}

\begin{proof} 
   Let $\I = \langle W,F,C, \V \rangle$ be the given instance, where $F=\{f_1,\dots,f_n\}$ and $W=\{1,\dots,m\}$ are the sets of firms and workers, respectively. 
   Let $\eta$ be the maximum possible Nash product for $\I$; thus, $\eta \leq (mv_{\max})^{m+n}$. Let $Z \coloneqq \{1,(1+\eps), (1+\eps)^2,\ldots,(1+\eps)^q,(1+\eps)^{q+1}\}$, where $q$ is the largest positive integer such that $(1+\eps)^q \leq \eta$. For every $j\in [n]$, $s\in [m]$, and $\ell \in Z$, we will construct a polynomial $p_{s,\ell}^j$ in which every nonzero monomial corresponds to an assignment of $s$ workers to the first $j$ firms $f_1,\dots,f_j$ such that the Nash product is at least $\ell$.   
   
   We will construct these polynomials $p_{s,\ell}^j$ iteratively. First, we will construct a polynomial $h_{s,\ell}^j$ in which every nonzero monomial corresponds to an assignment of some set $X \subseteq W$ of $s$ workers to the $j^\textup{th}$ firm $f_j$ such that $\W_{f_j}(X)\geq \ell$. 
   That is, for every $j\in [n]$, $s\in [c_j]$, and $\ell \in Z$,
   \begin{equation}
       h^j_{s,\ell} \coloneqq \sum_{\substack{X\subseteq W, |X|=s, \\ {\cal W}_{f_j}(X)\geq \ell}} y^{\chi(X)}.
   \label{eq:type1}
   \end{equation}

   Define $p^1_{s,\ell} \coloneqq h^1_{s,\ell}$. 
 %
%
 For every $j\in \{2,\ldots,n\}, s\in [m]$ and $\ell \in Z$, define
 \begin{equation}
       p^j_{s,\ell} \coloneqq {\cal R}\Big({\cal H}_s\Big(\sum_{\substack{s=s'+s'', s'\leq c_j, \\ \ell=\ell'\times \ell'', \ell',\ell''\in Z}} h_{s',\ell'}^j\times p_{s'',\ell''}^{j-1}\Big)\Big),
 \label{eq:typej}
 \end{equation}
where ${\cal R}(\cdot)$ is the representative polynomial and ${\cal H}_s(\cdot)$ is the Hamming projection of weight $s$. 
The ${\cal H}(\cdot)$ operator ensures disjointness due to~\Cref{cor:HW-disjointness}. The ${\cal R}(\cdot)$ operator is only required for the running time. 

We return the largest value of $\ell$ for which $p_{m,\ell}^n$ is nonzero. The corresponding matching can be found by backtracking. 

To argue correctness, we will prove that if the algorithm returns~$\ell^\star$, then $\eta \leq (1+\eps)^{n+1}\ell^\star$. Towards this, we show that if $(1+\eps)^q\leq \eta \leq (1+\eps)^{q+1}$, then $p^n_{m,\ell}$ is nonzero for some $\ell \in \{(1+\eps)^{q-n},\ldots, (1+\eps)^{q+1}\}$. Thus, we return a matching $\widetilde{\mu}$ with Nash product at least $(1+\eps)^{q-n}$. Since, $\eta \leq (1+\eps)^{q+1}$, it follows that $\eta \leq (1+\eps)^{n+1} \ell^\star$. Thus, $\W^\texttt{\textup{Nash}}(\mu) \leq (1+\eps)^{\frac{n+1}{m+n}}\W^\texttt{\textup{Nash}}(\widetilde{\mu})$, where $\mu$ is a Nash optimal matching. The detailed proof of correctness can be found in the appendix.

Since $|Z|=\OO(\log_{1+\eps} \eta)$, we compute $\OO(nm\log_{1+\eps} \eta)$ polynomials. Note that the degree of the polynomial 
is at most $2^m$ (as the the $m$-length binary vector in the polynomial can have all $1$s). Each polynomial can be computed in $\OO(m^2 \cdot 2^m \cdot \log \eta)$ time due to the following result and the fact that $s\leq m$ and $\ell \in Z$. 
\begin{restatable}[\citealp{moenck1976practical}]{proposition}{Moenck}
Two polynomials of degree $d$ can be multiplied in $\OO(d \log d)$ time.
\label{prop:polynomial-multiplication}
\end{restatable}
Since $\eta \leq (mv_{\max})^{m+n}$, $\log_{1+\eps} \eta \leq (m+n)\log_{1+\eps}(mv_{\max})$. By changing the base of logarithm from $(1+\eps)$ to $2$, and the fact that $\log(1+\eps)>\nicefrac{\eps^2}{2}$, for every $\eps \in (0,1]$, we get the running time of $\OO^\star(\nicefrac{1}{\eps^4} \cdot 2^m)$.
\end{proof}

\section{Restricted Domains}
\label{sec:Restricted_Domains}

In this section, we will design polynomial-time algorithms for restricted domains. 
%
%
%
Our first result is an algorithm for \emph{symmetric binary valuations}, which is when, for every worker-firm pair $(w,f) \in W \times F$, either $v_{w,f} = v_{f,w} = 1$ or $v_{w,f} = v_{f,w} = 0$. 

\begin{restatable}
{theorem}{SymmBinaryVals}
For symmetric binary valuations, a Nash optimal matching can be computed in polynomial time.
\label{thm:symm-Binary-Vals}
\end{restatable}

We use the algorithm of~\citet{BKV18greedy} which maximizes Nash welfare in the one-sided fair division problem under binary valuations.
%
Starting with a suboptimal allocation, their algorithm greedily picks a sequence of item transfers between agents to improve the Nash objective. We follow the same strategy, with the difference that the initial matching in our case is chosen using the algorithm in~\Cref{prop:Zero_Nash} to ensure nonzero Nash welfare. 

Note that for general symmetric valuations, the problem remains \NPH{}~(\Cref{prop:Nash-NPhard-Symmetric-Vals}). Whether an efficient algorithm is possible for (possibly asymmetric) binary valuations is an interesting problem for future work.

We will now discuss 
bounded degree instances.

\begin{restatable}[]{theorem}{Combined}
A Nash optimal matching can be computed in polynomial time when:
\begin{enumerate}
    \item all agents have degree at most $2$,
    \item all firms have degree at most $3$ and exactly two workers need to be assigned to every firm,
    \item every worker values only one firm positively, and
    \item the number of firms and the number of distinct valuations are constant.
\end{enumerate}
Furthermore, when the number of firms is constant and the number of distinct valuations is logarithmically bounded in the input size (i.e., at most $\log \left( \textup{poly}(m,n,v_{\max}) \right)$), a Nash optimal matching can be computed in quasipolynomial time.
\label{thm:Combined-Restricted-Domains}
\end{restatable}

\begin{proofsketch}
    \begin{enumerate}
        \item A degree $2$ instance is a disjoint union of paths and cycles. Each cycle and path can be solved in polynomial time. 
    \item 
    We first apply reduction rules 
    to ensure that any two firms can have only one worker in their common neighborhood. This reduces the problem to finding a weighted perfect matching, which can be solved efficiently~\citep{DBLP:conf/soda/Gabow90}.
    %
    \item Since every worker values only one firm positively, it can only be assigned to this firm. 
    \item The idea is similar to the  algorithm for \Cref{thm:QPTAS-Constant-Firms}. For each firm, we create a bucket for each valuation. 
    Since the number of distinct valuations and firms is constant, 
    the algorithm runs in polynomial time. The running time is quasipolynomial when the number of distinct valuations is logarithmically-bounded in the input size.\qedhere
    \end{enumerate}
\end{proofsketch}


The first three results in \Cref{thm:Combined-Restricted-Domains} provide tractable subcases vis-\`{a}-vis the \NPH{}ness result in \Cref{thm:Nash-NPhard-Capacity-Two} for degree at most $3$ and capacity $2$. The fourth result is in contrast with \NPH{}ness for two firms when the number of distinct valuations can be large (\Cref{prop:Nash-NPhard-Symmetric-Vals}).

\section{Concluding Remarks}
\label{sec:Concluding_Remarks}
We have initiated a systematic study of the computation of Nash optimal many-to-one matchings under two-sided preferences. Our work contributes a variety of hardness and algorithmic results, spanning a broad range of techniques including polynomial multiplication, submodular fair division, rainbow perfect matching, etc.~(see \Cref{tab:Results}).

A number of intriguing open problems remain. Designing a constant-factor approximation algorithm for our problem and developing lower bounds in terms of natural parameters are two outstanding questions. Additionally, studying Nash optimal solutions in conjunction with established solution concepts such as \emph{stability} and \emph{popularity} will also be of interest (our hardness result in \Cref{thm:Nash-NPhard-Capacity-Two} extends to these concepts). Finally, analogous to the literature in fair division, it would be interesting to find axiomatic justification for the use of Nash welfare in the two-sided setting.

\section*{Acknowledgments}
We are grateful to the anonymous reviewers of AAAI 2024 for helpful comments. RV acknowledges support from DST INSPIRE grant no. DST/INSPIRE/04/2020/000107
and SERB grant no. CRG/2022/002621.

\bibliographystyle{plainnat} 
\bibliography{References}

\begin{thebibliography}{70}
\providecommand{\natexlab}[1]{#1}
\providecommand{\url}[1]{\texttt{#1}}
\expandafter\ifx\csname urlstyle\endcsname\relax
  \providecommand{\doi}[1]{doi: #1}\else
  \providecommand{\doi}{doi: \begingroup \urlstyle{rm}\Url}\fi

\bibitem[Abdulkadiro{\u{g}}lu and S{\"o}nmez(2003)]{AS03school}
Atila Abdulkadiro{\u{g}}lu and Tayfun S{\"o}nmez.
\newblock {School Choice: A Mechanism Design Approach}.
\newblock \emph{American Economic Review}, 93\penalty0 (3):\penalty0 729--747,
  2003.

\bibitem[Akrami et~al.(2022)Akrami, Chaudhury, Hoefer, Mehlhorn, Schmalhofer,
  Shahkarami, Varricchio, Vermande, and van Wijland]{ACH+22maximizing}
Hannaneh Akrami, Bhaskar~Ray Chaudhury, Martin Hoefer, Kurt Mehlhorn, Marco
  Schmalhofer, Golnoosh Shahkarami, Giovanna Varricchio, Quentin Vermande, and
  Ernest van Wijland.
\newblock {Maximizing Nash Social Welfare in 2-Value Instances}.
\newblock In \emph{Proceedings of the 36th AAAI Conference on Artificial
  Intelligence}, volume~36, pages 4760--4767, 2022.

\bibitem[Amanatidis et~al.(2021)Amanatidis, Birmpas, Filos-Ratsikas, Hollender,
  and Voudouris]{ABF+21maximum}
Georgios Amanatidis, Georgios Birmpas, Aris Filos-Ratsikas, Alexandros
  Hollender, and Alexandros~A Voudouris.
\newblock {Maximum Nash Welfare and Other Stories about EFX}.
\newblock \emph{Theoretical Computer Science}, 863:\penalty0 69--85, 2021.

\bibitem[Amanatidis et~al.(2023)Amanatidis, Aziz, Birmpas, Filos-Ratsikas, Li,
  Moulin, Voudouris, and Wu]{AAB+23fair}
Georgios Amanatidis, Haris Aziz, Georgios Birmpas, Aris Filos-Ratsikas, Bo~Li,
  Herv{\'e} Moulin, Alexandros~A Voudouris, and Xiaowei Wu.
\newblock {Fair Division of Indivisible Goods: Recent Progress and Open
  Questions}.
\newblock \emph{Artificial Intelligence}, page 103965, 2023.

\bibitem[Anari et~al.(2018)Anari, Mai, Gharan, and Vazirani]{AMG+18nash}
Nima Anari, Tung Mai, Shayan~Oveis Gharan, and Vijay~V Vazirani.
\newblock {Nash Social Welfare for Indivisible Items under Separable,
  Piecewise-Linear Concave Utilities}.
\newblock In \emph{Proceedings of the 29th Annual ACM-SIAM Symposium on
  Discrete Algorithms}, pages 2274--2290, 2018.

\bibitem[Aziz et~al.(2023)Aziz, Freeman, Shah, and Vaish]{AFS+23best}
Haris Aziz, Rupert Freeman, Nisarg Shah, and Rohit Vaish.
\newblock {Best of Both Worlds: Ex-Ante and Ex-Post Fairness in Resource
  Allocation}.
\newblock \emph{Operations Research}, 2023.

\bibitem[Barman et~al.(2018{\natexlab{a}})Barman, Krishnamurthy, and
  Vaish]{BKV18finding}
Siddharth Barman, Sanath~Kumar Krishnamurthy, and Rohit Vaish.
\newblock {Finding Fair and Efficient Allocations}.
\newblock In \emph{Proceedings of the 2018 ACM Conference on Economics and
  Computation}, pages 557--574, 2018{\natexlab{a}}.

\bibitem[Barman et~al.(2018{\natexlab{b}})Barman, Krishnamurthy, and
  Vaish]{BKV18greedy}
Siddharth Barman, Sanath~Kumar Krishnamurthy, and Rohit Vaish.
\newblock {Greedy Algorithms for Maximizing Nash Social Welfare}.
\newblock In \emph{Proceedings of the 17th International Conference on
  Autonomous Agents and MultiAgent Systems}, pages 7--13, 2018{\natexlab{b}}.

\bibitem[Barman et~al.(2020)Barman, Bhaskar, Krishna, and
  Sundaram]{BBK+20tight}
Siddharth Barman, Umang Bhaskar, Anand Krishna, and Ranjani~G Sundaram.
\newblock {Tight Approximation Algorithms for $p$-Mean Welfare Under
  Subadditive Valuations}.
\newblock In \emph{Proceedings of the 28th Annual European Symposium on
  Algorithms}, 2020.

\bibitem[Barman et~al.(2021)Barman, Krishna, Kulkarni, and
  Narang]{BKK+21sublinear}
Siddharth Barman, Anand Krishna, Pooja Kulkarni, and Shivika Narang.
\newblock {Sublinear Approximation Algorithm for Nash Social Welfare with XOS
  Valuations}.
\newblock \emph{arXiv preprint arXiv:2110.00767}, 2021.

\bibitem[Bei et~al.(2019)Bei, Garg, Hoefer, and Mehlhorn]{BGH+19earning}
Xiaohui Bei, Jugal Garg, Martin Hoefer, and Kurt Mehlhorn.
\newblock {Earning and Utility Limits in Fisher Markets}.
\newblock \emph{ACM Transactions on Economics and Computation}, 7\penalty0
  (2):\penalty0 1--35, 2019.

\bibitem[Biswas and Barman(2018)]{BB18fair}
Arpita Biswas and Siddharth Barman.
\newblock {Fair Division Under Cardinality Constraints}.
\newblock In \emph{Proceedings of the 27th International Joint Conference on
  Artificial Intelligence}, pages 91--97, 2018.

\bibitem[Boehmer and Koana(2022)]{BK22complexity}
Niclas Boehmer and Tomohiro Koana.
\newblock {The Complexity of Finding Fair Many-To-One Matchings}.
\newblock In \emph{Proceedings of the 49th International Colloquium on
  Automata, Languages, and Programming}. Schloss Dagstuhl-Leibniz-Zentrum
  f{\"u}r Informatik, 2022.

\bibitem[Brams and Taylor(1996)]{BT96fair}
Steven~J Brams and Alan~D Taylor.
\newblock \emph{{Fair Division: From Cake-Cutting to Dispute Resolution}}.
\newblock Cambridge University Press, 1996.

\bibitem[Brandt et~al.(2016)Brandt, Conitzer, Endriss, Lang, and
  Procaccia]{BCE+16handbook}
Felix Brandt, Vincent Conitzer, Ulle Endriss, J{\'e}r{\^o}me Lang, and Ariel~D
  Procaccia.
\newblock \emph{{Handbook of Computational Social Choice}}.
\newblock Cambridge University Press, 2016.

\bibitem[Budish et~al.(2013)Budish, Che, Kojima, and Milgrom]{BCK+13designing}
Eric Budish, Yeon-Koo Che, Fuhito Kojima, and Paul Milgrom.
\newblock {Designing Random Allocation Mechanisms: Theory and Applications}.
\newblock \emph{American Economic Review}, 103\penalty0 (2):\penalty0 585--623,
  2013.

\bibitem[Caragiannis et~al.(2019)Caragiannis, Kurokawa, Moulin, Procaccia,
  Shah, and Wang]{CKM+19unreasonable}
Ioannis Caragiannis, David Kurokawa, Herv{\'e} Moulin, Ariel~D Procaccia,
  Nisarg Shah, and Junxing Wang.
\newblock {The Unreasonable Fairness of Maximum Nash Welfare}.
\newblock \emph{ACM Transactions on Economics and Computation}, 7\penalty0
  (3):\penalty0 1--32, 2019.

\bibitem[Chaudhury et~al.(2021)Chaudhury, Garg, and Mehta]{CGM21fair}
Bhaskar~Ray Chaudhury, Jugal Garg, and Ruta Mehta.
\newblock {Fair and Efficient Allocations under Subadditive Valuations}.
\newblock In \emph{Proceedings of the Thirty-Fifth AAAI Conference on
  Artificial Intelligence}, volume~35, pages 5269--5276, 2021.

\bibitem[Chaudhury et~al.(2022)Chaudhury, Cheung, Garg, Garg, Hoefer, and
  Mehlhorn]{CCG+22fair}
Bhaskar~Ray Chaudhury, Yun~Kuen Cheung, Jugal Garg, Naveen Garg, Martin Hoefer,
  and Kurt Mehlhorn.
\newblock {Fair Division of Indivisible Goods for a Class of Concave
  Valuations}.
\newblock \emph{Journal of Artificial Intelligence Research}, 74:\penalty0
  111--142, 2022.

\bibitem[Cole and Gkatzelis(2018)]{CG18approximating}
Richard Cole and Vasilis Gkatzelis.
\newblock {Approximating the Nash Social Welfare with Indivisible Items}.
\newblock \emph{SIAM Journal on Computing}, 47\penalty0 (3):\penalty0
  1211--1236, 2018.

\bibitem[Cole et~al.(2017)Cole, Devanur, Gkatzelis, Jain, Mai, Vazirani, and
  Yazdanbod]{CDG+17convex}
Richard Cole, Nikhil Devanur, Vasilis Gkatzelis, Kamal Jain, Tung Mai, Vijay~V
  Vazirani, and Sadra Yazdanbod.
\newblock {Convex Program Duality, Fisher Markets, and Nash Social Welfare}.
\newblock In \emph{Proceedings of the 2017 ACM Conference on Economics and
  Computation}, pages 459--460, 2017.

\bibitem[Cook et~al.(2021)Cook, Diamond, Hall, List, and Oyer]{CDH+21gender}
Cody Cook, Rebecca Diamond, Jonathan~V Hall, John~A List, and Paul Oyer.
\newblock {The Gender Earnings Gap in the Gig Economy: Evidence from Over a
  Million Rideshare Drivers}.
\newblock \emph{The Review of Economic Studies}, 88\penalty0 (5):\penalty0
  2210--2238, 2021.

\bibitem[Cygan et~al.(2015)Cygan, Fomin, Kowalik, Lokshtanov, Marx, Pilipczuk,
  Pilipczuk, and Saurabh]{ParamAlgorithms15b}
M.~Cygan, F.~Fomin, L.~Kowalik, D.~Lokshtanov, D.~Marx, M.~Pilipczuk,
  M.~Pilipczuk, and S.~Saurabh.
\newblock \emph{{Parameterized Algorithms}}.
\newblock Springer, 2015.

\bibitem[Cygan and Pilipczuk(2010)]{DBLP:journals/tcs/CyganP10}
Marek Cygan and Marcin Pilipczuk.
\newblock {Exact and Approximate Bandwidth}.
\newblock \emph{Theoretical Computer Science}, 411\penalty0 (40-42):\penalty0
  3701--3713, 2010.

\bibitem[Darmann and Schauer(2015)]{DS15maximizing}
Andreas Darmann and Joachim Schauer.
\newblock {Maximizing Nash Product Social Welfare in Allocating Indivisible
  Goods}.
\newblock \emph{European Journal of Operational Research}, 247\penalty0
  (2):\penalty0 548--559, 2015.

\bibitem[Downey and Fellows(2013)]{downey}
R.~G. Downey and M.~R. Fellows.
\newblock \emph{{Fundamentals of Parameterized Complexity}}.
\newblock Springer-Verlag, 2013.

\bibitem[Dror et~al.(2023)Dror, Feldman, and Segal-Halevi]{DFS23fair}
Amitay Dror, Michal Feldman, and Erel Segal-Halevi.
\newblock {On Fair Division under Heterogeneous Matroid Constraints}.
\newblock \emph{Journal of Artificial Intelligence Research}, 76:\penalty0
  567--611, 2023.

\bibitem[Eisenberg and Gale(1959)]{EG59consensus}
Edmund Eisenberg and David Gale.
\newblock {Consensus of Subjective Probabilities: The Pari-Mutuel Method}.
\newblock \emph{The Annals of Mathematical Statistics}, 30\penalty0
  (1):\penalty0 165--168, 1959.

\bibitem[Feder(1995)]{F95stable}
Tom{\'a}s Feder.
\newblock \emph{{Stable Networks and Product Graphs}}, volume 555.
\newblock American Mathematical Soc., 1995.

\bibitem[Fisher et~al.(1978)Fisher, Nemhauser, and Wolsey]{FNW78analysis}
Marshall~L Fisher, George~L Nemhauser, and Laurence~A Wolsey.
\newblock {An Analysis of Approximations for Maximizing Submodular Set
  Functions--II}.
\newblock \emph{Mathematical Programming Studies}, 8:\penalty0 73--87, 1978.

\bibitem[Flum and Grohe(2006)]{fg}
J.~Flum and M.~Grohe.
\newblock \emph{{Parameterized Complexity Theory}}.
\newblock Springer-Verlag, 2006.

\bibitem[Freeman et~al.(2021)Freeman, Micha, and Shah]{FMS21twosided}
Rupert Freeman, Evi Micha, and Nisarg Shah.
\newblock {Two-Sided Matching Meets Fair Division}.
\newblock In \emph{Proceedings of the 30th International Joint Conference on
  Artificial Intelligence}, pages 203--209, 2021.

\bibitem[Gabow(1990)]{DBLP:conf/soda/Gabow90}
Harold~N. Gabow.
\newblock {Data Structures for Weighted Matching and Nearest Common Ancestors
  with Linking}.
\newblock In \emph{Proceedings of the First Annual {ACM-SIAM} Symposium on
  Discrete Algorithms}, pages 434--443, 1990.

\bibitem[Gale and Shapley(1962)]{GS62college}
David Gale and Lloyd~S Shapley.
\newblock {College Admissions and the Stability of Marriage}.
\newblock \emph{The American Mathematical Monthly}, 69\penalty0 (1):\penalty0
  9--15, 1962.

\bibitem[Garey and Johnson(1979)]{GJ79computers}
Michael~R Garey and David~S Johnson.
\newblock \emph{{Computers and Intractability}}, volume 174.
\newblock W. H. Freeman and Company, 1979.

\bibitem[Garg et~al.(2018)Garg, Hoefer, and Mehlhorn]{GHM18approximating}
Jugal Garg, Martin Hoefer, and Kurt Mehlhorn.
\newblock {Approximating the Nash Social Welfare with Budget-Additive
  Valuations}.
\newblock In \emph{Proceedings of the Twenty-Ninth Annual ACM-SIAM Symposium on
  Discrete Algorithms}, pages 2326--2340, 2018.

\bibitem[Garg et~al.(2020)Garg, Kulkarni, and Kulkarni]{GKK20approximating}
Jugal Garg, Pooja Kulkarni, and Rucha Kulkarni.
\newblock {Approximating Nash Social Welfare under Submodular Valuations
  through (Un) Matchings}.
\newblock In \emph{Proceedings of the Fourteenth Annual ACM-SIAM Symposium on
  Discrete Algorithms}, pages 2673--2687, 2020.

\bibitem[Garg et~al.(2022)Garg, Husic, Murhekar, and V\'{e}gh]{GHM22tractable}
Jugal Garg, Edin Husic, Aniket Murhekar, and L\'{a}szl\'{o} V\'{e}gh.
\newblock {Tractable Fragments of the Maximum Nash Welfare Problem}.
\newblock In \emph{Proceedings of the 18th International Conference on Web and
  Internet Economics}, volume 13778, page 362, 2022.

\bibitem[Garg et~al.(2023)Garg, Husi{\'c}, Li, V{\'e}gh, and
  Vondr{\'a}k]{GHL+22approximating}
Jugal Garg, Edin Husi{\'c}, Wenzheng Li, L{\'a}szl{\'o}~A V{\'e}gh, and Jan
  Vondr{\'a}k.
\newblock {Approximating Nash Social Welfare by Matching and Local Search}.
\newblock In \emph{Proceedings of the 55th Annual ACM Symposium on Theory of
  Computing}, pages 1298--1310, 2023.

\bibitem[Gupta et~al.(2023)Gupta, Nagori, Chakraborty, Vaish, Ranu, Nadkarni,
  Dasararaju, and Chelliah]{GNC+23towards}
Anjali Gupta, Shreyans~J Nagori, Abhijnan Chakraborty, Rohit Vaish, Sayan Ranu,
  Prajit~Prashant Nadkarni, Narendra~Varma Dasararaju, and Muthusamy Chelliah.
\newblock {Towards Fair Allocation in Social Commerce Platforms}.
\newblock In \emph{Proceedings of the ACM Web Conference 2023}, pages
  3744--3754, 2023.

\bibitem[Gupta et~al.(2021{\natexlab{a}})Gupta, Jain, Saurabh, and
  Talmon]{DBLP:conf/ijcai/Gupta00T21}
Sushmita Gupta, Pallavi Jain, Saket Saurabh, and Nimrod Talmon.
\newblock Even more effort towards improved bounds and fixed-parameter
  tractability for multiwinner rules.
\newblock In \emph{{Proceedings of the Thirtieth International Joint Conference
  on Artificial Intelligence}}, pages 217--223, 2021{\natexlab{a}}.

\bibitem[Gupta et~al.(2021{\natexlab{b}})Gupta, Roy, Saurabh, and
  Zehavi]{GRS+21balanced}
Sushmita Gupta, Sanjukta Roy, Saket Saurabh, and Meirav Zehavi.
\newblock {Balanced Stable Marriage: How Close is Close Enough?}
\newblock \emph{Theoretical Computer Science}, 883:\penalty0 19--43,
  2021{\natexlab{b}}.

\bibitem[Gusfield and Irving(1989)]{GI89stable}
Dan Gusfield and Robert~W Irving.
\newblock \emph{{The Stable Marriage Problem: Structure and Algorithms}}.
\newblock MIT press, 1989.

\bibitem[Hitsch et~al.(2010)Hitsch, Horta{\c{c}}su, and Ariely]{HHA10matching}
G{\"u}nter~J Hitsch, Ali Horta{\c{c}}su, and Dan Ariely.
\newblock {Matching and Sorting in Online Dating}.
\newblock \emph{American Economic Review}, 100\penalty0 (1):\penalty0 130--163,
  2010.

\bibitem[Hummel and Hetland(2022)]{HH22maximin}
Halvard Hummel and Magnus~Lie Hetland.
\newblock {Maximin Shares under Cardinality Constraints}.
\newblock In \emph{Proceedings of the 19th European Conference on Multi-Agent
  Systems}, pages 188--206. Springer, 2022.

\bibitem[Igarashi et~al.(2022)Igarashi, Kawase, Suksompong, and
  Sumita]{IKS+22fair}
Ayumi Igarashi, Yasushi Kawase, Warut Suksompong, and Hanna Sumita.
\newblock {Fair Division with Two-Sided Preferences}.
\newblock \emph{arXiv preprint arXiv:2206.05879}, 2022.

\bibitem[Irving et~al.(1987)Irving, Leather, and Gusfield]{ILG87efficient}
Robert~W Irving, Paul Leather, and Dan Gusfield.
\newblock {An Efficient Algorithm for the “Optimal” Stable Marriage}.
\newblock \emph{Journal of the ACM}, 34\penalty0 (3):\penalty0 532--543, 1987.

\bibitem[Kaneko and Nakamura(1979)]{KN79nash}
Mamoru Kaneko and Kenjiro Nakamura.
\newblock {The Nash Social Welfare Function}.
\newblock \emph{Econometrica: Journal of the Econometric Society}, pages
  423--435, 1979.

\bibitem[Knuth(1997)]{K97stable}
Donald~Ervin Knuth.
\newblock \emph{{Stable Marriage and its Relation to Other Combinatorial
  Problems: An Introduction to the Mathematical Analysis of Algorithms}},
  volume~10.
\newblock American Mathematical Soc., 1997.

\bibitem[Lee(2017)]{L17apx}
Euiwoong Lee.
\newblock {APX-Hardness of Maximizing Nash Social Welfare with Indivisible
  Items}.
\newblock \emph{Information Processing Letters}, 122:\penalty0 17--20, 2017.

\bibitem[Manlove(2013)]{M13algorithmics}
David Manlove.
\newblock \emph{{Algorithmics of Matching under Preferences}}, volume~2.
\newblock World Scientific, 2013.

\bibitem[McVitie and Wilson(1971)]{MW71stable}
David~G McVitie and Leslie~B Wilson.
\newblock {The Stable Marriage Problem}.
\newblock \emph{Communications of the ACM}, 14\penalty0 (7):\penalty0 486--490,
  1971.

\bibitem[Moenck(1976)]{moenck1976practical}
R.~T. Moenck.
\newblock {Practical Fast Polynomial Multiplication}.
\newblock In \emph{Proceedings of the Third ACM Symposium on Symbolic and
  Algebraic Computation}, pages 136--148, 1976.

\bibitem[Moulin(2004)]{M04fair}
Herv{\'e} Moulin.
\newblock \emph{{Fair Division and Collective Welfare}}.
\newblock MIT press, 2004.

\bibitem[Narang et~al.(2022)Narang, Biswas, and Narahari]{NBN22achieving}
Shivika Narang, Arpita Biswas, and Yadati Narahari.
\newblock {On Achieving Leximin Fairness and Stability in Many-to-One
  Matchings}.
\newblock In \emph{Proceedings of the 21st International Conference on
  Autonomous Agents and Multiagent Systems}, pages 1705--1707, 2022.

\bibitem[Nash~Jr(1950)]{N50bargaining}
John~F Nash~Jr.
\newblock {The Bargaining Problem}.
\newblock \emph{Econometrica: Journal of the Econometric Society}, pages
  155--162, 1950.

\bibitem[Nguyen et~al.(2014)Nguyen, Nguyen, Roos, and
  Rothe]{NNR+14computational}
Nhan-Tam Nguyen, Trung~Thanh Nguyen, Magnus Roos, and J{\"o}rg Rothe.
\newblock {Computational Complexity and Approximability of Social Welfare
  Optimization in Multiagent Resource Allocation}.
\newblock \emph{Autonomous Agents and Multi-Agent Systems}, 28\penalty0
  (2):\penalty0 256--289, 2014.

\bibitem[Niedermeier(2006)]{DBLP:books/ox/Niedermeier06}
Rolf Niedermeier.
\newblock \emph{{Invitation to Fixed-Parameter Algorithms}}.
\newblock Oxford University Press, 2006.

\bibitem[Oxley(2022)]{O22matroid}
James Oxley.
\newblock {Matroid Theory}.
\newblock In \emph{Handbook of the Tutte Polynomial and Related Topics}, pages
  44--85. Chapman and Hall/CRC, 2022.

\bibitem[Robertson and Webb(1998)]{RW98cake}
Jack Robertson and William Webb.
\newblock \emph{{Cake-Cutting Algorithms: Be Fair if You Can}}.
\newblock CRC Press, 1998.

\bibitem[Roth and Peranson(1999)]{RP99redesign}
Alvin~E Roth and Elliott Peranson.
\newblock {The Redesign of the Matching Market for American Physicians: Some
  Engineering Aspects of Economic Design}.
\newblock \emph{American Economic Review}, 89\penalty0 (4):\penalty0 748--780,
  1999.

\bibitem[Roth and Sotomayor(1992)]{RS92two}
Alvin~E Roth and Marilda Sotomayor.
\newblock {Two-Sided Matching}.
\newblock \emph{Handbook of Game Theory with Economic Applications},
  1:\penalty0 485--541, 1992.

\bibitem[Roth et~al.(2004)Roth, S{\"o}nmez, and {\"U}nver]{R04kidney}
Alvin~E Roth, Tayfun S{\"o}nmez, and M~Utku {\"U}nver.
\newblock {Kidney Exchange}.
\newblock \emph{The Quarterly Journal of Economics}, 119\penalty0 (2):\penalty0
  457--488, 2004.

\bibitem[Segal-Halevi(2019)]{S19fair}
Erel Segal-Halevi.
\newblock {Fair Division with Bounded Sharing}.
\newblock \emph{arXiv preprint arXiv:1912.00459}, 2019.

\bibitem[Sethuraman et~al.(2006)Sethuraman, Teo, and Qian]{STQ06many}
Jay Sethuraman, Chung-Piaw Teo, and Liwen Qian.
\newblock {Many-To-One Stable Matching: Geometry and Fairness}.
\newblock \emph{Mathematics of Operations Research}, 31\penalty0 (3):\penalty0
  581--596, 2006.

\bibitem[Shoshan et~al.(2023)Shoshan, Hazon, and Segal-Halevi]{SHS23efficient}
Hila Shoshan, Noam Hazon, and Erel Segal-Halevi.
\newblock {Efficient Nearly-Fair Division with Capacity Constraints}.
\newblock In \emph{Proceedings of the 2023 International Conference on
  Autonomous Agents and Multiagent Systems}, pages 206--214, 2023.

\bibitem[Suksompong(2021)]{S21constraints}
Warut Suksompong.
\newblock {Constraints in Fair Division}.
\newblock \emph{ACM SIGecom Exchanges}, 19\penalty0 (2):\penalty0 46--61, 2021.

\bibitem[Tanimoto et~al.(1978)Tanimoto, Itai, and Rodeh]{TIR78some}
Steven~L Tanimoto, Alon Itai, and Michael Rodeh.
\newblock {Some Matching Problems for Bipartite Graphs}.
\newblock \emph{Journal of the ACM}, 25\penalty0 (4):\penalty0 517--525, 1978.

\bibitem[Teo and Sethuraman(1998)]{TS98geometry}
Chung-Piaw Teo and Jay Sethuraman.
\newblock {The Geometry of Fractional Stable Matchings and its Applications}.
\newblock \emph{Mathematics of Operations Research}, 23\penalty0 (4):\penalty0
  874--891, 1998.

\bibitem[Vazirani(2007)]{V07combinatorial}
Vijay~V Vazirani.
\newblock {Combinatorial Algorithms for Market Equilibria}.
\newblock In \emph{Algorithmic Game Theory}, pages 103--134. 2007.

\end{thebibliography}

\clearpage
\appendix
\begin{center}
   \LARGE{\textbf{Appendix}}
\end{center}

\section{Additional Related Work}
\label{subsec:Additional_Related_Work}

In this section, we will survey the related literature. We have divided the discussion into two parts, based on whether the literature deals with \emph{one-sided} (\Cref{subsec: Related-Work-One-sided}) or \emph{two-sided} preferences (\Cref{subsec: Related-Work-Two-sided}).

\subsection{Fairness under One-Sided Preferences}
\label{subsec: Related-Work-One-sided}

The literature on fairness under one-sided preferences, wherein a set of agents have preferences over a set of resources, is covered under the area of \emph{fair division}. The broad goal here is to divide the resources among the agents subject to satisfying some fairness property~(such as envy-freeness or proportionality) and/or optimizing some welfare measure~(such as utilitarian, Nash, or egalitarian welfare).

\paragraph{Divisible resources.} When the resources are \emph{divisible}, the problem of fair division is known under the name of \emph{cake cutting}~\citep{BT96fair,RW98cake,BCE+16handbook}. The special case of \emph{homogeneous} resources, wherein an agent derives equal utility from equal amounts of a given resource,  has received substantial interest within economic theory. An important result in this literature is that a fractional allocation that maximizes Nash social welfare induces a market equilibrium; in particular, such an allocation is envy-free and Pareto optimal. Further, a maximum Nash welfare allocation can be computed in polynomial time using a convex program~\citep{EG59consensus,V07combinatorial}.

\paragraph{Indivisible resources.} For \emph{indivisible} resources, the picture is more involved. An envy-free allocation may no longer exist, and as such a Nash optimal allocation fails to satisfy this property (though it continues to be Pareto optimal). Further, maximizing Nash welfare is \APXH{} under additive valuations~\citep{NNR+14computational,L17apx,GHM18approximating}. We refer the reader to \citep{AAB+23fair} for a survey on fair division with indivisible resources.

\paragraph{Algorithmic results.} For indivisible resources under additive valuations, a fully polynomial-time approximation scheme (FPTAS) is known for any fixed number of agents~\citep{NNR+14computational,GHM22tractable}. When the number of agents can be arbitrary, there are constant-factor approximation algorithms known~\citep{CG18approximating,CDG+17convex,BKV18finding}. Several works have studied approximation algorithms for more general classes of valuations, such as \emph{budget-additive}~\citep{GHM18approximating}, \emph{separable piecewise-linear concave}~\citep{AMG+18nash,CCG+22fair}, \emph{submodular}~\citep{GKK20approximating,GHL+22approximating}, \emph{XOS}~\citep{BKK+21sublinear}, and \emph{subadditive}~\citep{BBK+20tight,CGM21fair}. 

\paragraph{Local search algorithms.} 
Some of the approximation algorithms for Nash welfare are based on \emph{local search}~\citep{BKV18finding,GHL+22approximating}. Here, starting with a suboptimal allocation, the algorithm searches for a ``nearby'' allocation that improves upon certain objective (e.g., improving the Nash welfare or reducing the envy among the agents), and repeats until no further (significant) improvement can be made. One might expect these ideas to generalize to the two-sided setting. However, as the example in \Cref{fig:Motivating_Example} suggests, improving Nash welfare in the two-sided instance might come at the expense of worsening the welfare in its ``one-sided projection''. Thus, the techniques from the one-sided preferences literature may not straightforwardly apply to the two-sided problem.

\paragraph{Binary valuations.} For fair division under binary valuations, it is known that a Nash optimal allocation can be efficiently computed~\citep{DS15maximizing,BKV18greedy}. Our work shows a similar tractability result for \emph{symmetric binary valuations}~(\Cref{thm:symm-Binary-Vals}) in the two-sided matching problem, and leaves open the resolution of the case with (possibly asymmetric) binary valuations.

\paragraph{Fair division under cardinality constraints.}
Our model generalizes the problem of fair division with indivisible items in two ways: First, in our model, the `items' (equivalently, workers) also have preferences over the agents (i.e., the firms). Secondly, each agent has a constraint on the maximum number of items it can receive. By dropping the first but keeping the second generalization, we obtain the problem of \emph{fair division under cardinality constraints} as a special case of our model~\citep{BB18fair}. 

The literature on fair division under cardinality constraints has focused primarily on axiomatic properties such as (approximate) envy-freeness, maximin share, and Pareto optimality~\citep{BB18fair,HH22maximin,SHS23efficient}. However, the computational aspects of Nash welfare in this model have not been studied as extensively. An exception is the recent work of \citet{GNC+23towards}, who consider a model with two-sided cardinality constraints and one-sided preferences. In their work, each agent must be assigned at least (respectively, at most) a certain number of items. Additionally, for each item, at least (respectively, at most) a certain number of its copies must be distributed among the agents. The agents have preferences over the items but not vice versa. The goal is to find a Nash optimal allocation among all cardinality-constrained allocations. This problem is shown to be \NPH{} even when each agent receives exactly three items and there is exactly one copy of each item.

Nash welfare has also been studied under other types of constraints such as spending restrictions and earning limits in a market-based setting~\citep{CG18approximating,BGH+19earning}. We refer the reader to the article by \citet{DFS23fair} and the survey by \citet{S21constraints} for a more comprehensive overview of constraints in fair division.

\subsection{Fairness under Two-Sided Preferences}
\label{subsec: Related-Work-Two-sided}

Let us now turn our attention to the literature on fairness in the two-sided matching problem.

\paragraph{Fairness in the stable matching problem.} The literature on two-sided matchings originated from the seminal work of \citet{GS62college} on the stable matching problem. The deferred acceptance algorithm proposed in this work has had a remarkable theoretical and practical impact~\citep{RS92two}. However, the algorithm is also known to return solutions that strongly favor one side of the market at the expense of the other~\citep{GS62college,MW71stable}.

The unfairness of deferred acceptance algorithm, together with the immense practical applicability of stable matchings, has generated considerable interest in developing algorithms for finding \emph{fair} stable matchings. Several fairness concepts have been studied in conjunction with stability, including \emph{minimum regret}~\citep{K97stable}, \emph{egalitarian}~\citep{MW71stable,ILG87efficient},\footnote{In the stable matching literature, the term \emph{egalitarian matching} has been used for matchings that maximize the ``average satisfaction'' of the agents by minimizing the \emph{sum} of ranks of the matched partners~\citep{MW71stable,ILG87efficient}. This objective is different from \emph{egalitarian welfare}, which maximizes the utility of the least-happy agent. If the rank of matched partner of an agent is interpreted as its \emph{disutility} or \emph{cost}, then the objective in the egalitarian matching problem, in fact, turns out to be \emph{utilitarian} welfare.} \emph{median}~\citep{TS98geometry,STQ06many}, \emph{sex-equal}~\citep{GI89stable}, \emph{balanced}~\citep{F95stable,GRS+21balanced}, and \emph{leximin}~\citep{NBN22achieving}. 
%
%
%

The work of \citet{NBN22achieving}, in particular, is conceptually similar to ours. They study a many-to-one matching model under cardinal preferences (similar to our work), and analyze the leximin fairness criterion which lexicographically maximizes the utilities of agents, i.e., among all matchings, it considers those that maximize the utility of the minimum-utility agent, subject to which the second-highest utility is maximized, and so on. They show that finding a leximin-optimal stable matching is \NPH{} in general, but there is an efficient algorithm for \emph{ranked} valuations wherein agents have the same preference ordering over the other side. The leximin solution concept generalizes egalitarian welfare which only seeks to maximize the utility of the least-happy agent.

\paragraph{Fairness in other two-sided matching problems.} Some recent works have studied fairness in two-sided matchings beyond the stability requirement. \citet{FMS21twosided} study fairness for many-to-many matchings under two-sided preferences. Like our work, their fairness goals are also inspired from the fair division literature. However, the exact notions considered in their work differ from ours; specifically, they focus on approximate envy-freeness and maximin share. \citet{IKS+22fair} study many-to-one matchings satisfying approximate envy-freeness and balancedness properties. They additionally considers variants of the stability condition such as swap stability or individual stability.

Finally, \citet{BK22complexity} study a many-to-one matching problem without preferences. They consider a bipartite graph in which the vertices on the ``one'' side (i.e., the side with unit capacity) each have a color, and the goal is to find a matching that assigns to each agent on the ``many'' side an almost equal number of vertices of each color.


\section{Omitted Details from \Cref{sec:Preliminaries}}

\subsection{Additional Preliminaries}

\paragraph{Fixed-parameter tractable (\fpt{}) algorithm.} In parameterized algorithms, given an instance $I$ of a problem $\Pi$ and an integer $k$ (also known as the \emph{parameter}), the goal is to design an algorithm that runs in $f(k)\cdot |I|^{\OO(1)}$ time, where $f$ is an
arbitrary computable function depending on the parameter $k$. Such an algorithm is said to be \emph{fixed-parameter tractable} (\fpt{}) in the parameter $k$. 
For more details, 
we refer the reader to the texts by~\citet{ParamAlgorithms15b},~\citet{fg},~\citet{downey} and~\citet{DBLP:books/ox/Niedermeier06}. 

\paragraph{\fpt{} approximation scheme (FPT-AS).}
A \emph{fixed-parameter tractable approximation scheme} (FPT-AS)\footnote{Not to be confused with a fully polynomial-time approximation scheme (\textup{FPTAS}).} 
is an approximation algorithm that, given an instance $I$ of a maximization problem $\Pi$, an integer $k$, and any $\eps > 0$, returns a $\frac{1}{(1+\eps)}$-approximate solution in $f(k)\cdot |I|^{\OO(1)}$ time, where $f$ is an
arbitrary computable function depending on the parameter $k$. 

\subsection{Proof of \Cref{prop:Zero_Nash}}


Recall the statement of \Cref{prop:Zero_Nash}.

\ZeroNash*

The proof of \Cref{prop:Zero_Nash} uses two main technical ideas: First, it runs a linear program on a \emph{binarized} instance obtained by suitably adjusting the original valuations. Second, it uses the rounding technique of~\citet{BCK+13designing} (and its strengthening by~\citet{AFS+23best}) to show that the said linear program is feasible if and only if there exists a matching with nonzero Nash welfare in the \emph{original} instance. Thus, checking the existence of a matching with nonzero Nash welfare reduces to checking the feasibility of a linear program, which can be done in polynomial time.

Let us now describe the key notions that will be used in the proof of \Cref{prop:Zero_Nash}, starting with that of a \emph{binarized} instance. Given an input instance $\I = \langle W,F,C,\V \rangle$, the corresponding binarized instance $\I' = \langle W,F,C,\V' \rangle$ consists of the same set of workers and firms as the original instance $\I$ but with different valuations. Specifically, for every worker~$w$ and every firm $f \in F$, we have:
$$v'_{w,f} = 1 \text{ if and only if } v_{w,f} > 0, \text{ and }$$
$$v'_{f,w} = 1 \text{ if and only if } v_{f,w} > 0.$$
Otherwise, we set $v'_{w,f}$ and $v'_{f,w}$ to be $0$.

The binarized instance ``flattens'' out the preferences of the agents by replacing all nonzero valuations in the original instance $\I$ by~1. As we will see shortly, doing so helps in setting up our linear program in accordance with the framework of~\citet{BCK+13designing}. Note that a matching has positive Nash welfare in the original instance $\I$ if and only if it has positive Nash welfare in the binarized instance $\I'$. Thus, it suffices to focus only on the binarized instance~$\I'$.

Next, let us describe the linear program that we will run on the binarized instance. Consider the following linear program, which we call $\texttt{LP}$, that checks whether there is a \emph{fractional} many-to-one matching in the binarized instance under which each agent gets a nonzero utility:
\begin{alignat*}{2}
    \protect \text{maximize} \ \ \ & \ \ \ \ \ \ 0 & \ \ & \\
    \text{subject to} \ \ \ & \sum_{f \in F} \mu(w,f) \leq 1 && \forall w \in W,\\
    & \sum_{w \in W} \mu(w,f)  \leq c_f && \forall f \in F,\\
    & \sum_{f \in F} v'_{w,f} \cdot \mu(w,f)  \geq 1 && \forall w \in W,\\
    & \sum_{w \in W} v'_{f,w} \cdot \mu(w,f)  \geq 1 && \forall f \in F,\\
    & \mu(w,f) \geq 0 && \forall w \in W \text{ and } \forall f \in F.
\end{alignat*}
The first two constraints ensure that no more than one unit of each worker and no more than $c_f$ units of the firm $f$ are matched. The third and the fourth constraints ensure that, in the binarized instance, the utility of each worker and each firm is positive. Since the binarized instance has binary valuations, this condition is equivalent to asking that each agent's utility is at least 1. The fifth constraint ensures non-negativity of the matching variables.

The final technical ingredient that we will need is the rounding technique of~\citet{BCK+13designing}. We will describe their framework in the context of our linear program $\texttt{LP}$. Our exposition will closely follow that of~\citet{AFS+23best}.

Let $\mu$ be any fractional solution of $\texttt{LP}$. Consider the matrix representation of $\mu$ where the workers correspond to rows and the firms correspond to columns. The first constraint, namely $\sum_{f \in F} \mu(w,f) \leq 1 \text{ for all } w \in W$, can be seen as a ``row-wise'' constraint. Similarly, the second constraint that applies to each firm can be seen as ``column-wise''. More generally, we can define a \emph{constraint set} $S \in 2^{W \times F} $ consisting of a collection of worker-firm pairs. A \emph{constraint structure} $\mathcal{H}$ consists of all distinct constraint sets associated with $\texttt{LP}$.

A constraint structure is said to be a \emph{hierarchy} (or a laminar family) if any pair of constraint sets in it are either nested or disjoint. That is, for every $S,S' \in \mathcal{H}$, either $S \subset S'$ or $S' \subset S$ or $S \cap S' = \emptyset$. A constraint structure is said to be a \emph{bihierarchy} if it can be partitioned into two hierarchies, i.e., $\mathcal{H} = \mathcal{H}_1 \cup \mathcal{H}_2$ such that both $\mathcal{H}_1$ and $\mathcal{H}_2$ are hierarchies and $\mathcal{H}_1 \cap \mathcal{H}_2 = \emptyset$.

\citet{BCK+13designing} considered problems in which there is a quota constraint associated with each constraint set. Specifically, for each constraint set $S$, there is a \emph{lower quota} $\underline{q}_S$ and an \emph{upper quota} $\overline{q}_S$ such that
$$\underline{q}_S \leq \sum_{(w,f) \in S} \mu(w,f) \leq \overline{q}_S.$$

It is easy to check that our linear program \texttt{LP} has a bihierarchical constraint structure with quota constraints. Indeed, for every worker, the first, third, and fifth constraints together constitute one hierarchy. Similarly, for every firm, the second and fourth constraints together constitute the second hierarchy. The two hierarchies are mutually disjoint, and therefore together they constitute a bihierarchical constraint structure. Further, since the binarized instance has binary valuations, the constraints can be written in terms of lower and upper quotas.

A fractional solution $\mu$ of $\texttt{LP}$ is said to admit a \emph{feasible decomposition} if there exist integral many-to-one matchings $\mu^1,\mu^2,\dots,\mu^\ell$ such that:
\begin{itemize}
    \item $\mu$ can be written as a convex combination of $\mu^1,\mu^2,\dots,\mu^\ell$, i.e., $\mu = \sum_{k \in [\ell]} p_k \mu^k$ where $p_1,p_2,\dots,p_\ell$ are nonnegative numbers satisfying $\sum_{k \in [\ell]} p_k = 1$, and
    \item each $\mu^k$ is feasible, i.e., for every $k \in [\ell]$, we have $\underline{q}_S \leq \sum_{(w,f) \in S} \mu^k(w,f) \leq \overline{q}_S$.
\end{itemize}

\citet{BCK+13designing} showed that any fractional solution that satisfies bihierarchical quota constraints admits a feasible decomposition. (Note that this result already generalizes the well-known Birkhoff-von Neumann decomposition.) Furthermore, they showed that under any integral matching in the decomposition, the utility of each agent is ``not too far'' from its utility in the fractional matching. This result was strengthened by~\citet{AFS+23best} in the context of the fair division problem with \emph{one-sided} preferences. In \Cref{prop:Utility_Guarantee} below, we adapt the result of~\citet{AFS+23best} to the \emph{two-sided} matching problem. (The proof of \Cref{prop:Utility_Guarantee} follows by making straightforward modifications to the argument of~\citet{AFS+23best} and is therefore omitted.) We will write $u'_w$ and $u'_f$ to denote the utility of worker $w$ and firm $f$, respectively, in the binarized instance $\I'$.

\begin{restatable}[Adapted from~\citealp{AFS+23best}]{proposition}{UtilityGuarantee}
There is a polynomial-time algorithm that, given any fractional matching $\mu$ of the linear program \texttt{LP}, computes a feasible decomposition of $\mu$ into integral matchings $\mu^1,\mu^2,\dots,\mu^\ell$ such that for every $k \in [\ell]$, we have:
\begin{itemize}
    \item for any worker $w$, the utility of $w$ in the integral matching $\mu^k$ is equal to $1$, i.e., $u'_w(\mu^k) = 1$,
    \item for any firm $f$, if $u'_f(\mu^k) < u'_f(\mu)$, then there exists a worker $w \notin \mu^k(f)$ such that $\mu(w,f) > 0$ and $u'_f(\mu^k) + u'_f(\{w\}) > u'_f(\mu)$, and
    \item for any firm $f$, if $u'_f(\mu^k) > u'_f(\mu)$, then there exists a worker $w \in \mu^k(f)$ such that $\mu(w,f) < 1$ and $u'_f(\mu^k) - u'_f(\{w\}) < u'_f(\mu)$.
\end{itemize}
\label{prop:Utility_Guarantee}
\end{restatable}


We are now ready to prove \Cref{prop:Zero_Nash}.

\begin{proof} (of \Cref{prop:Zero_Nash})
Our algorithm is the following: Return YES if the execution of \texttt{LP} returns a feasible (possibly fractional) solution, otherwise return NO. It is easy to see that the algorithm runs in polynomial time. Thus, we will proceed to showing its correctness.

Suppose there exists an integral matching in the binarized instance $\I'$ in which all agents derive positive utility, i.e., each agent has utility at least 1. Then such a matching will satisfy the linear program \texttt{LP}. Thus, it is safe for our algorithm to return YES.


We will now argue the converse: Whenever \texttt{LP} is feasible (equivalently, whenever our algorithm returns YES), there exists an integral matching in the binarized instance in which all agents derive positive utility.

Suppose $\mu$ is a feasible fractional solution of \texttt{LP}. Since \texttt{LP} has a bihierarchical constraint structure with quota constraints, we can apply the utility guarantee in \Cref{prop:Utility_Guarantee}.

From the first condition in~\Cref{prop:Utility_Guarantee}, it follows that in each integral matching in the decomposition of $\mu$, each worker's utility is equal to 1. From the second condition, it follows that if, in some integral matching $\mu^k$, the utility of firm $f$ is less than its utility in the fractional matching $\mu$, then the difference between the two utility values is \emph{strictly less} than $1$ (recall that the valuations are binary). We know that by the feasibility of $\mu$, the utility of firm $f$ under $\mu$ is at least $1$. Thus, it follows that the utility of firm $f$ under $\mu^k$ is strictly positive, as desired.
\end{proof}

As a corollary of \Cref{prop:Zero_Nash}, we obtain an algorithm that also efficiently identifies instances with nonzero \emph{egalitarian} welfare (defined as the utility of the least-happy agent).

\section{Omitted Proofs from \Cref{sec:Hardness_Results}}

This section will provide the complete proofs of the results mentioned in \Cref{sec:Hardness_Results} on the hardness of computing a Nash optimal matching. We will start by proving \Cref{lem:RainbowPerfectMatching-NPhard} which shows \NPH{}ness of a restricted version of the \RainbowPerfectMatching{} problem. Next, we will use this result to show \NPH{}ness of maximizing Nash welfare when each firm has capacity 2~(\Cref{thm:Nash-NPhard-Capacity-Two}). Finally, in \Cref{prop:Nash-NPhard-Symmetric-Vals}, we will use a reduction from \Partition{} to show \NPH{}ness of maximizing Nash welfare even for two identical firms and even under symmetric valuations.

\subsection{Proof of \NPH{}ness of Rainbow Perfect Matching}

The \RainbowPerfectMatching{} is defined as follows: The input to this problem consists of a bipartite multigraph $G = (X \cup Y,E)$. The vertex sets $X$ and $Y$ each consist of $r$ vertices (i.e., $|X| = |Y| = r$), and the edges in $E$ are partitioned into $r$ color classes $E_1,E_2,\dots,E_r$. Between any two vertices, there can be multiple edges but at most one edge of a given color. The goal is to determine if there exists a perfect matching in $G$ that has exactly one edge from each color class.

This problem is known to be \NPC{} even when each color class consists of at most two edges; see~\citep[Problem GT55]{GJ79computers} and~\citep{TIR78some}. However, the maximum degree of the graph in this construction can be large. In \Cref{lem:RainbowPerfectMatching-NPhard}, we show that the problem remains \NPH{} even when the maximum degree is three, albeit at the expense of having three (instead of two) edges in each color class. The degree bound is essential in showing the bound on the number of positively valued agents in~\Cref{thm:Nash-NPhard-Capacity-Two}.

\begin{restatable}{lemma}{}
\RainbowPerfectMatching{} is \NPC{} even for graphs where every vertex has degree three and there are exactly three edges of each color.
\label{lem:RainbowPerfectMatching-NPhard}
\end{restatable}
\begin{proof}
We will show a reduction from \ThreeDPerfectMatching{} which is defined as follows: An instance of \ThreeDPerfectMatching{} consists of a tripartite graph $G = (X \cup Y \cup Z,E)$, where the vertex sets $X$, $Y$, and $Z$ have the same cardinality, i.e., $|X|=|Y|=|Z|=r$. Each edge in $E$ consists of a triple of vertices, one each from $X$, $Y$, and $Z$, i.e., $E \subseteq X \times Y \times Z$. The goal is to determine if there is a perfect 3D matching of $G$, i.e., if there is a subset $S \subseteq E$ such that each vertex in $X \cup Y \cup Z$ is adjacent to exactly one edge in $S$. It is known that \ThreeDPerfectMatching{} remains \NPC{} even when each vertex in $X \cup Y \cup Z$ is adjacent to exactly three edges in $E$~\cite[SP1]{GJ79computers}. We will use this restricted version of \ThreeDPerfectMatching{} in our proof. We will find it helpful to index the vertices in the set $Z$ as $z_1,z_2,\dots,z_r$.

Given a \ThreeDPerfectMatching{} instance with a graph $G = (X \cup Y \cup Z, E)$, we will construct an instance of \RainbowPerfectMatching{} as follows: Construct a bipartite multigraph $G' = (X \cup Y,E')$ where the edge set $E'$ is the union of $r$ color classes $E_1,E_2,\dots,E_r$. For every $i \in [r]$, the color class $E_i$ consists of all ordered pairs $(x,y) \in X \times Y$ such that $(x,y,z_i) \in E$, i.e., $E_i \coloneqq \{(x,y) \in X \times Y : (x,y,z_i) \in E\}$. We can assume that there are no repeated elements in $E$. Thus, between any pair of vertices $(x,y) \in X \times Y$, there is at most one edge of color $i$. Furthermore, since each vertex in the \ThreeDPerfectMatching{} instance has degree three, it follows that each vertex in \RainbowPerfectMatching{} instance also has degree three and that there are three edges of each color~(corresponding to each $z_i \in Z$).

We will now argue that \ThreeDPerfectMatching{} admits a YES instance if and only if \RainbowPerfectMatching{} does.

($\Rightarrow$) Suppose the \ThreeDPerfectMatching{} instance admits a solution $\mu$, i.e, $\mu$ is a perfect 3D matching in graph $G$. Consider a matching $\mu' \coloneqq \{(x,y) \in X \times Y: (x,y,z_i) \in E \text{ for some } z_i \in Z \}$ in the multigraph $G'$. Then, $\mu'$ is a valid perfect matching in $G'$ since each vertex in $X \cup Y$ is adjacent to exactly one edge in $\mu'$. Furthermore, all edges in $\mu'$ must have distinct colors since each edge in the 3D matching $\mu$ is adjacent to a distinct vertex in the set $Z$. Thus, the matching $\mu'$ is a valid rainbow perfect matching.

($\Leftarrow$) Suppose the \RainbowPerfectMatching{} instance admits a solution $\mu'$. Consider a 3D matching $\mu \coloneqq \{ (x,y,z_i) \in X \times Y \times Z: (x,y) \in \mu' \text{ has color } i\}$. Each vertex in $X \cup Y$ is adjacent to exactly one edge in $\mu'$, and the same holds for the 3D matching $\mu$ in the graph $G$. Furthermore, since all edges in $\mu'$ have distinct colors, each vertex in $Z$ is adjacent to exactly one edge in $\mu$. Thus, $\mu$ is the desired 3D perfect matching.
\end{proof}

\subsection{Proof of \Cref{thm:Nash-NPhard-Capacity-Two}}


\NPHardCapacityTwo*

\begin{proof} (of \Cref{thm:Nash-NPhard-Capacity-Two})
We will show a reduction from \RainbowPerfectMatching{}. Recall that an instance of \RainbowPerfectMatching{} consists of a bipartite multigraph $G = (X \cup Y,E)$. The vertex sets $X$ and $Y$ each consist of $r$ vertices (i.e., $|X| = |Y| = r$), and the edges in $E$ are partitioned into $r$ color classes $E_1,E_2,\dots,E_r$. Between any two vertices, there can be multiple edges but at most one edge of a given color. The goal is to determine if there exists a perfect matching in $G$ that has exactly one edge from each color class. From \Cref{lem:RainbowPerfectMatching-NPhard}, we know that \RainbowPerfectMatching{} remains \NPC{} even when the degree of every vertex in graph $G$ is three and there are three edges in each color class. 
We will use this restricted version of \RainbowPerfectMatching{} in our proof. Thus, in particular, we will assume that the edge set $E$ consists of $3r$ edges, three of each color.

Given a \RainbowPerfectMatching{} instance with a graph $G = (X \cup Y,E)$, we will construct a two-sided matching instance as follows: We have $2r$ \emph{main} workers $w_1,\dots,w_{2r}$, each corresponding to a distinct vertex in $X \cup Y$, and $3r$ \emph{main} firms, each corresponding to an edge in $E$. For every $i \in [r]$, there are three main firms, namely $f^i_1,f^i_2,f^i_3$, corresponding to the three edges in color class $E_i$.

In addition, there are $r$ \emph{dummy} firms $g_1,\dots,g_r$ and $3r$ \emph{dummy} workers. The dummy workers are partitioned into $r$ classes where each class corresponds to a distinct color. For every $i \in [r]$, the $i^\text{th}$ dummy class consists of three dummy workers, namely $d^i_1,d^i_2,d^i_3$.

We will write $N \coloneqq 3r + r + 2r + 3r = 9r$ to denote the total number of agents (i.e, main firms, dummy firms, main workers, and dummy workers).

\begin{figure}[t]
    \centering
    \small
    \begin{tikzpicture}[scale=0.84]
                \tikzset{firm/.style = {shape=rectangle,draw,inner sep=2.5pt}}
                \tikzset{dummyfirm/.style = {shape=rectangle,fill=black,draw,inner sep=2.5pt}}
                \tikzset{worker/.style = {shape=circle,draw,inner sep=1.5pt}}
                \tikzset{dummyworker/.style = {shape=circle,draw,fill=black,inner sep=1.5pt}}
                \tikzset{edge/.style = {solid}}
                \draw[draw=black] (0,3.5) rectangle ++(.5,1.4);
                \node (1) at (-2,4.2) {\color{black}{Class 1 main firms}};
                \node[] (0) at (0.25,3.3) {$\vdots$};
                \draw[draw=black] (0,1.5) rectangle ++(.5,1.4);
                \node (1) at (-2,2.2) {\color{black}{Class $r$ main firms}};
                \draw[draw=black] (0,0) rectangle ++(.5,1.4);
                \node (1) at (-2,0.7) {\color{black}{$r$ dummy firms}};
                \draw[draw=black] (2,0) rectangle ++(.5,1.4);
                \node (1) at (5.25,0.5) {\color{black}{Class $r$ dummy workers}};
                \node[] (0) at (2.25,1.8) {$\vdots$};
                \draw[draw=black] (2,2) rectangle ++(.5,1.4);
                \node (1) at (5.25,2.5) {\color{black}{Class 1 dummy workers}};
                \draw[draw=black] (2,3.9) rectangle ++(.5,1.6);
                \node (1) at (4.5,4.65) {\color{black}{$2r$ main workers}};
                \node[dummyfirm] (1) at (0.25,0.3) {};
                \node[] (0) at (0.25,0.8) {$\vdots$};
                \node[dummyfirm] (3) at (0.25,1.1) {};
                \node[firm] (4) at (0.25,1.85) {};
                \node[firm] (5) at (0.25,2.2) {};
                \node[firm] (6) at (0.25,2.55) {};
                \node[firm] (7) at (0.25,3.85) {};
                \node[firm] (8) at (0.25,4.2) {};
                \node[firm] (9) at (0.25,4.55) {};
                \node[dummyworker] (10) at (2.25,0.33) {};
                \node[dummyworker] (11) at (2.25,0.66) {};
                \node[dummyworker] (12) at (2.25,1) {};
                \node[dummyworker] (13) at (2.25,2.33) {};
                \node[dummyworker] (14) at (2.25,2.66) {};
                \node[dummyworker] (15) at (2.25,3) {};
                \node[worker] (16) at (2.25,4.2) {};
                \node[] (0) at (2.25,4.8) {$\vdots$};
                \node[worker] (18) at (2.25,5.2) {};
                \draw[edge] (9) to node [near start,fill=white,inner sep=0pt] (9218) {\scriptsize{$1$}} node [near end,fill=white,inner sep=0pt] (1829) {\scriptsize{$1$}} (18);
                \draw[edge] (9) to node [near start,fill=white,inner sep=0pt] (9216) {\scriptsize{$1$}} node [near end,fill=white,inner sep=0pt] (1629) {\scriptsize{$1$}} (16);
                \draw[edge] (9) to node [near start,fill=white,inner sep=0pt] (8215) {\scriptsize{$2$}} node [near end,fill=white,inner sep=0pt] (1528) {\scriptsize{$1$}} (15);
                \draw[edge] (1) to node [near start,fill=white,inner sep=0pt] (1211) {\scriptsize{$2$}} node [near end,fill=white,inner sep=0pt] (1121) {\scriptsize{$1$}} (11);
                \draw[edge] (3) to node [near start,fill=white,inner sep=0pt] (3213) {\scriptsize{$2$}} node [near end,fill=white,inner sep=0pt] (1323) {\scriptsize{$1$}} (13);
                \draw[edge] (7) to node [near start,fill=white,inner sep=0pt] (7213) {\scriptsize{$2$}} node [near end,fill=white,inner sep=0pt] (1327) {\scriptsize{$1$}} (13);
    \end{tikzpicture}
    \repeatcaption{fig:Nash-NPhard-Capacity-Two}{The reduced two-sided matching instance in the proof of \Cref{thm:Nash-NPhard-Capacity-Two}. Firms are denoted by squares and workers by circles. The shaded and unshaded nodes denote the dummy and the main agents, respectively.}
\end{figure}

The preferences of the agents are as follows (see \Cref{fig:Nash-NPhard-Capacity-Two}): 
\begin{itemize}
    \item Each main firm corresponding to an edge $e = (u,v)$ in graph $G$ values two main workers---namely, those corresponding to the vertices $u$ and $v$---each at $1$, and has value $0$ for the remaining main workers.

    \item For every $i \in [r]$ and $j \in \{1,2,3\}$, the main firm $f^i_j$ (i.e., the $j^\text{th}$ firm in color class $i$) values the dummy worker $d^i_j$ (i.e., $j^\text{th}$ dummy worker in color class $i$) at $2$, and has value $0$ for all other dummy workers. We will call $d^i_j$ the \emph{signature} dummy worker for the main firm $f^i_j$.

    \item For every $i \in [r]$, each dummy firm $g_i$ values the three dummy workers in the $i^\text{th}$ color class, namely $d^i_1,d^i_2,d^i_3$, at $2$, and has value $0$ for all other workers (main and dummy).

    \item Each main worker corresponding to a vertex $v$ in the graph $G$ values each firm corresponding to an edge $e \in E$ such that $v \in e$ at $1$ and the remaining main and dummy firms at $0$.

    \item For every $i \in [r]$, each of the three dummy workers in color class $i$, namely $d^i_1,d^i_2,d^i_3$, value the dummy firm $g_i$ at $1$, and all other main and dummy firms at $0$.
\end{itemize}

Observe that the preferences constructed above have a ``zero value symmetry'' structure, in that a firm $f$ values a worker $w$ at $0$ if and only if the worker $w$ values the firm $f$ at $0$. Therefore, if there is a matching with nonzero Nash social welfare, then all workers must have strictly positive utility, implying that no zero-valued edge is selected in such a matching.

The capacity of each firm is 2.

We will use the threshold $\theta \coloneqq 2^{\# \text{firms}/N} = 2^{4/9}$ to specify the decision version of the matching problem. That is, our goal is to determine if there is a matching with Nash social welfare at least $\theta$.

We will now argue that \RainbowPerfectMatching{} admits a YES instance if and only if the two-sided matching problem does.\newline

$(\Rightarrow)$ Suppose the \RainbowPerfectMatching{} instance admits a solution $\mu$, i.e., $\mu$ is a perfect matching in graph $G$ and all edges in $\mu$ have distinct colors. Consider a matching $\mu'$ in the reduced instance as follows: 
\begin{itemize}
    \item For each edge $e \in \mu$ such that $e = (u,v)$, we assign the main workers $u$ and $v$ to the main firm $e$. Since the matching $\mu$ has rainbow property, all main workers are assigned after this step and exactly one main firm from each color class is used.
    \item Next, for each edge $e \notin \mu$ (note that there are $2r$ such edges), we will assign to each main firm $e$ its signature dummy worker. This leaves exactly one signature dummy worker in each color class unassigned.
    \item The unassigned dummy worker in each color class is assigned to the corresponding dummy firm.
\end{itemize}

Note that the matching $\mu'$ assigns every worker (main or dummy) to exactly one firm, and each firm is assigned at most two workers. Thus, $\mu'$ is feasible.

Furthermore, each firm derives a utility of exactly $2$. Indeed, a main firm either gets two main workers or one dummy worker, and each dummy firm gets exactly one dummy worker. Similarly, each worker (main or dummy), by virtue by being assigned along a nonzero value edge, derives a utility of exactly $1$. There are $4r$ firms and $5r$ workers; thus $N = 9r$ agents overall. Therefore, the Nash welfare of $\mu'$ is $2^{4r/9r}$, which is equal to $\theta$, as desired.\newline

$(\Leftarrow)$ Suppose there is a matching $\mu'$ with Nash welfare at least $\theta > 0$. Thus, in particular, every agent (worker and firm) derives positive utility under $\mu'$.

By construction, if a worker (main or dummy) derives positive utility, it must have a utility of $1$. Thus, the Nash welfare of $\mu'$ must be $(\nicefrac{1}{N})^\text{th}$ power of the product of the utilities of the firms, where $N=9r$ is the total number of agents (workers and firms). In other words, the product of firms' utilities under $\mu'$ must be at least $2^{4r}$. Recall that the number of firms is $4r$. Thus, the geometric mean of firms' utilities under $\mu'$ is at least $2$.

We will now argue that under $\mu'$, the arithmetic mean of firms' utilities is exactly $2$. By the AM-GM inequality, this would imply that each firm's utility under $\mu'$ is exactly $2$, which, in turn, will give us the desired assignment, as we will show later.\footnote{If the arithmetic mean of a set of numbers equals their geometric mean, then the numbers under consideration must all be equal.}

To see why the arithmetic mean of firms' utilities is exactly $2$, observe that under $\mu'$, each dummy worker (of color class $i$) is either assigned to the $i^\text{th}$ dummy firm or to its signature main firm. In each case, each dummy worker contributes $2$ to the sum of firms' utilities. Similarly, any main worker $v$, corresponding to vertex $v$ in graph $G$, must be assigned to a main firm $e$ such that $v \in e$ in graph $G$, contributing $1$ to the sum of firms' utilities (see \Cref{fig:Nash-NPhard-Capacity-Two}). 

There are $3r$ dummy workers and $2r$ main workers. Therefore, the sum of firms' utilities must be $2\cdot 3r + 1 \cdot 2r = 8r$. There are $4r$ firms in total, and therefore the arithmetic mean of firms' utilities is $8r / 4r = 2$. As the arithmetic and the geometric mean of firms' utilities under $\mu'$ turn out to be equal, it must be that under $\mu'$, each firm's utility is \emph{exactly} $2$.

We will now use the fact that each firm's utility is exactly $2$ to infer the desired rainbow perfect matching. Note that the only way a dummy firm, say $g_i$, has utility $2$ is if it is assigned exactly one dummy worker out of $d^i_1,d^i_2,d^i_3$. This leaves two dummy workers out of $d^i_1,d^i_2,d^i_3$ to be assigned to the main firms.

Once again, recall that each agent is assigned to a firm that it has a nonzero value for. Thus, each dummy worker (in color class $i$) that is not assigned to a dummy firm must be assigned to its signature main firm. The said main firm will then have a utility of $2$ and therefore cannot be assigned any main workers that it values. This leaves $r$ main firms---one from each color class---that each derive a utility of $2$ exclusively from the $2r$ main agents. This is only possible if each main firm (corresponding to an edge $e = (u,v)$) is assigned the main workers corresponding to the vertices $u$ and $v$.

The above assignment naturally induces a perfect matching in the \RainbowPerfectMatching{} instance. Finally, since exactly one main firm from each color class is assigned two main workers, this matching also has the rainbow property, as desired.
\end{proof}

\subsection{Proof of \Cref{prop:Nash-NPhard-Symmetric-Vals}}


\SymmetricValsNPhard*

\begin{proof}(of \Cref{prop:Nash-NPhard-Symmetric-Vals})
We will show a reduction from \Partition{} which is known to be \NPH{}~\citep{GJ79computers}. An instance of \Partition{} consists of $m$ positive integers $a_1,\dots,a_m$. The goal is to determine if there exists a subset $S \subseteq [m]$ such that $\sum_{i \in S} a_i = \sum_{k \in [m] \setminus S} a_k$. Let $T \coloneqq \frac{1}{2} \sum_{i \in [m]} a_i$ denote the target sum. It is known that \Partition{} remains \NPC{} even when the given integers $a_1,\dots,a_m$ are distinct and the desired partitions are equisized, i.e., $|S| = m/2$~\citep[Theorem B.1]{S19fair}. In our reduction, we will assume the \Partition{} instance to have these properties.

Given a \Partition{} instance with positive integers $a_1,\dots,a_m$, we will construct a matching instance as follows: We will have $m$ workers $w_1,\dots,w_m$ and two firms $f_1,f_2$. Each firm $f \in \{f_1,f_2\}$ values the worker $w_i$ at $a_i$, and the worker $w_i$ values each firm at $a_i$; thus, the valuations are symmetric. Additionally, both firms have \emph{identical} valuations. The capacity of each firm is set to $c_1 = c_2 = m/2$.

We will use the threshold $\theta \coloneqq \left(T^2 \cdot \Pi_{i \in [m]} a_i \right)^\frac{1}{m+2}$ to specify the decision version of the matching problem. That is, our goal is to determine if there is a matching with Nash social welfare at least $\theta$.

We will now argue that the \Partition{} instance admits a YES instance if and only if the matching instance does.

$(\Rightarrow)$ Suppose the \Partition{} instance admits a solution $S$. Consider a matching $\mu$ where the workers corresponding to numbers in the set $S$ are matched with firm $f_1$ and the remaining workers are matched with firm $f_2$. The utility derived by worker $w_i$ is $a_i$, and the utility derived by firm $f_1$ (respectively, $f_2$) is the sum of numbers in the set $S$ (respectively, $[m] \setminus S$). By the partition property, each firm derives a utility equal to $T$. It is easy to check that the Nash welfare of this matching is equal to $\theta$, as desired.

$(\Leftarrow)$ Suppose there is a matching $\mu$ with Nash welfare at least $\theta$. Observe that, without loss of generality, no worker is unmatched under $\mu$ (otherwise, the Nash welfare can be weakly improved by assigning each unmatched worker to one of the firms). The utility derived by worker $w_i$ under $\mu$ is $a_i$. Thus, the product of workers' utilities under $\mu$ is $\Pi_{i \in [m]} a_i$.

For the matching $\mu$ to have Nash social welfare at least $\theta$, it must be that the product of firms' utilities under $\mu$ is at least $T^2$, i.e., $u_{f_1}(\mu) \cdot u_{f_2}(\mu) \geq T^2$. In other words, the geometric mean of firms' utilities must be \emph{at least}~$T$.

Since the firms' utilities are additive, we have that $u_{f_1}(\mu) + u_{f_2}(\mu) = \sum_{i \in \mu(f_1)} a_i + \sum_{k \in \mu(f_2)} a_k = \sum_{i \in [m]} a_i = 2T$. That is, the arithmetic mean of firms' utilities under $\mu$ is \emph{exactly} $T$.

By the AM-GM inequality, the geometric mean of a set of numbers is at most their arithmetic mean, and the two quantities are equal only when the numbers are all equal. Thus, it must be that $u_{f_1}(\mu) = u_{f_2}(\mu) = T$. Furthermore, since the capacity of each firm is $m/2$, we have that $|\mu(f_1)| = |\mu(f_2)| = m/2$.

Let $S \coloneqq \{i \in [m] : w_i \in \mu(f_1)\}$. It follows from the above argument that $\sum_{i \in S} a_i = T$ and $|S|=m/2$, thus inducing the desired balanced partition.
\end{proof}

The proof of \Cref{prop:Nash-NPhard-Symmetric-Vals} can be modified to show that computing a maximum Nash welfare matching is \NPC{} even when both workers and firms have \emph{strict} preferences (i.e., no ties), as follows: The worker $w_i$ values firm $f_1$ at $a_i$ and firm $f_2$ at $2 a_i$. The firm $f_1$ values worker $w_i$ at $a_i$ while the firm $f_2$ values it at $a_i/2^{m/2}$. As before, the capacity of each firm is $m/2$ and the threshold Nash welfare is $\theta$.

\StrictPrefsNPhard*

\section{Omitted Proofs from \Cref{sec:Approximation_Algorithms}}

\subsection{Proof of \Cref{thm:Nash-Approx-OPT-Dependent}}

Recall the statement of \Cref{thm:Nash-Approx-OPT-Dependent}.

\NashApproxOPTdependent*

To prove \Cref{thm:Nash-Approx-OPT-Dependent}, we will show that the task of designing an approximation algorithm for Nash welfare in the two-sided matching problem reduces to designing an approximation algorithm for \emph{utilitarian} welfare in the one-sided \emph{fair division} problem under \emph{submodular} valuations and \emph{matroid} constraints. 

Specifically, we will construct a fair division instance with monotone submodular valuations where the firms play the role of \emph{agents} and the workers play the role of \emph{items}; the detailed construction is described below. We will then show the following correspondence between \emph{allocations} in the fair division instance and \emph{induced matchings} in the original matching instance: The utilitarian welfare of an allocation in the fair division instance is equal to a scaled version of the log of Nash welfare of the induced matching in the original instance~(\Cref{prop:Utilitarian_Nash_Relationship}). This observation will allow us to use a known approximation algorithm for maximizing a submodular function under matroid constraints~(\Cref{prop:Submodular-Utilitarian-Approximation}) to provide approximation guarantees for Nash welfare in the matching problem.

Let us start by describing the construction of the associated fair division instance. Given any two-sided matching instance $\I = \langle W,F,C, \V \rangle$, we will construct a fair division instance $\I' = \langle [n], [m], \hat{\V} \rangle$ consisting of $n=|F|$ agents $a_1,a_2,\dots,a_n$ (wherein agent $a_i$ corresponds to firm $f_i$), $m=|W|$ items $g_1,g_2,\dots,g_m$ (where item $g_j$ corresponds to worker $w_j$), and a valuation profile $\hat{\V} = (\hat{v}_1,\hat{v}_2,\dots,\hat{v}_n)$ of \emph{modified} valuation functions for the agents, as defined below.

\begin{definition}[Modified valuation function]
For any fixed firm $f \in F$, let $a \in [n]$ be its corresponding agent in the fair division instance. Fix any set $X \subseteq W$ of workers, and let $Y \in [m]$ denote the corresponding set of items. The \emph{modified} valuation function of agent $a$ in the fair division instance is given by $\hat{v}_a: 2^{[m]} \rightarrow \mathbb{R}_{\geq 0}$, where, the value of agent $a$ for the items in $Y$ is defined as
$$\hat{v}_a(Y) \coloneqq \log \left( v_f(X) \cdot \prod_{w \in X} v_{w,f} \right).$$
Furthermore, $\hat{v}_a(\emptyset) \coloneqq 0$.\qed
%
\label{def:Modified_Valuation}
\end{definition}

An \emph{allocation} $A = (A_1,\dots,A_n)$ in the fair division instance is an $n$-partition of the items among the agents such that $\cup_{i \in [n]} A_i = [m]$ and for any pair of distinct agents $i,h \in [n]$, $A_i \cap A_h = \emptyset$.

Given any allocation $A$ in the fair division instance, one can naturally associate with it a matching in the two-sided matching instance, which we will call the \emph{induced} matching, as defined below.

\begin{definition}[Induced matching]
Given an allocation $A$, its \emph{induced matching} $\mu$ is defined as follows: If the item $g_j$ is assigned to agent $a_i$ in the fair division instance $\I'$, then the corresponding worker $w_j$ is assigned to the corresponding firm $f_i$ under $\mu$ in the original matching instance $\I$. That is,
$$ g_j \in A_i \Rightarrow w_j \in \mu(f_i).$$\qed
\label{def:Induced_Matching}
\end{definition}

\paragraph{Value queries.} The modified valuation function is defined over all $2^m$ subsets of items. In order to ensure that the size of the input to our algorithm in the fair division instance is not exponential in $m$, we will assume that the valuation function can be accessed via \emph{value queries}. This query assumes oracle access to the value $v_i(S)$ of any given subset of items $S \subseteq [m]$ of items and any agent $i \in [n]$. Observe that the form of the modified valuation function in \Cref{def:Modified_Valuation} allows us to simulate any such query in polynomial time.

The specific form of modified valuations in \Cref{def:Modified_Valuation} gives rise to the following relationship between the utilitarian welfare of an allocation and the Nash welfare of its induced matching.

\begin{restatable}{proposition}{}
Let $A$ be any allocation in the fair division instance $\I'$ and let $\mu$ be the matching induced by $A$ in the original instance $\I$. Then, the utilitarian welfare of $A$ is a scaled version of the log of Nash welfare of $\mu$, i.e.,
$$\W^\texttt{\textup{util}}(A) = (n+m) \log \W^\texttt{\textup{Nash}}(\mu).$$
\label{prop:Utilitarian_Nash_Relationship}
\end{restatable}
\begin{proof} (of \Cref{prop:Utilitarian_Nash_Relationship}) The utilitarian welfare of allocation $A$ is given by
$$\W^\texttt{\textup{util}}(A) = \sum_{i \in [n]} \hat{v}_i(A_i).$$
By definition of modified valuation function~(\Cref{def:Modified_Valuation}) and induced matching $\mu$~(\Cref{def:Induced_Matching}), we have that
$$\W^\texttt{\textup{util}}(A) = \sum_{i \in [n]} \log \left( v_i(\mu(i)) \cdot \prod_{j \in \mu(i)} v_{j,i} \right),$$
where $v_i(\mu(i))$ is firm $f_i$'s valuation for the set of workers~$\mu(i)$ corresponding to the items in the set $A_i$.

The right hand side in the above equality can be rewritten as
\begin{align*}
    \sum_{i \in [n]} & \log \left( v_i(\mu(i)) \cdot \prod_{j \in \mu(i)} v_{j,i} \right) \\
    & = \log \left( \prod_{i \in [n]} \left( v_i(\mu(i)) \cdot \prod_{j \in \mu(i)} v_{j,i} \right) \right) \\
    & = (n+m) \cdot \log \left( \prod_{i \in [n]} \left( v_i(\mu(i)) \cdot \prod_{j \in \mu(i)} v_{j,i} \right) \right)^{\nicefrac{1}{n+m}} \\
    & = (n+m) \cdot \log \W^\texttt{\textup{Nash}}(\mu).
\end{align*}
\end{proof}

We will now show that the modified valuation function $\hat{v}$ is monotone and submodular~(\Cref{prop:Modified_Valuation_Submodular}). 
This observation requires the valuations to be strictly positive.

\begin{restatable}{proposition}{}
The modified valuation function $\hat{v}$ is monotone and submodular.
\label{prop:Modified_Valuation_Submodular}
\end{restatable}
\begin{proof} (of \Cref{prop:Modified_Valuation_Submodular}) Let us start by proving monotonicity. Consider any pair of subsets of items $S, T \subseteq [m]$ such that $S \subseteq T$. We will show that for any agent $i \in [n]$, $\hat{v}_i(S) \leq \hat{v}_i(T)$.

Let $f$ denote the firm corresponding to agent $i \in [n]$. For any subset of items $T \subseteq [m]$, we will overload notation to denote the set of corresponding workers also by $T$.

The desired condition is trivially satisfied if $S = \emptyset$. So, let us assume that $S$ is nonempty. By definition of modified valuation function~(\Cref{def:Modified_Valuation}), we have that
\begin{align*}
    \hat{v}_i(T) & = \log \left( v_f(T) \cdot \prod_{w \in T} v_{w,f} \right) \\
    & = \log \left( v_f(T) \cdot \prod_{w \in S} v_{w,f} \prod_{w \in T \setminus S} v_{w,f} \right) \\
    & = \log \left( ( v_f(S) + v_f(T \setminus S) ) \cdot \prod_{w \in S} v_{w,f} \prod_{w \in T \setminus S} v_{w,f} \right) \\
    & \geq \log \left( ( v_f(S) ) \cdot \prod_{w \in S} v_{w,f} \right),
\end{align*}
where the third equality uses additivity of firms' valuations, and the  inequality uses the facts that the firms' valuations are nonnegative and the workers in $T \setminus S$ have positive valuations. Throughout, the $\log(\cdot)$ terms are well-defined because $S$ is nonempty and firms' valuations are assumed to be strictly positive. This establishes that $\hat{v}_i$ is monotone.

Let us now argue that $\hat{v}_i$ is submodular. That is, we need to prove the following: For any subsets of items $S, T \subseteq [m]$ such that $S \subseteq T$ and for any item $j \notin T$, we have
$$\hat{v}_i(T \cup \{j\}) - \hat{v}_i(T) \leq \hat{v}_i(S \cup \{j\}) - \hat{v}_i(S).$$

If $T = \emptyset$, then the condition follows trivially. So, let us assume that $T$ is nonempty. From~\Cref{def:Modified_Valuation}, it follows that
\begin{align*}
\hat{v}_i&(T \cup \{j\}) - \hat{v}_i(T) \\
& = \log \left( v_f(T \cup \{j\}) \cdot \prod_{w \in T \cup \{j\}} v_{w,f} \right) - \\
& \hspace{1.5in} \log \left( v_f(T) \cdot \prod_{w \in T} v_{w,f} \right) \\
& = \log \left( \frac{v_f(T \cup \{j\})}{v_f(T)} \cdot v_{j,f} \right) \\
& = \log \left( \left( 1 + \frac{v_{f,j}}{v_f(T)} \right) \cdot v_{j,f} \right),
\end{align*}
where the third equality uses additivity of firms' valuations, and, as before, the $\log(\cdot)$ terms are well-defined because of positive valuations.

Similarly, we get that
$$\hat{v}_i(S \cup \{j\}) - \hat{v}_i(S) = \log \left( \left( 1 + \frac{v_{f,j}}{v_f(S)} \right) \cdot v_{j,f} \right).$$

Since the firms' valuation are nonnegative and additive, we have that $v_f(S) \leq v_f(T)$. This readily gives the desired condition for submodularity of $\hat{v}_i$.    
\end{proof}

Next, we will discuss the transfer of approximation guarantees between the fair division and matching instances.

\begin{restatable}{proposition}{}
Let $A$ be an $\alpha$-utilitarian optimal allocation in the fair division instance with modified valuation functions $\hat{v}_1,\dots,\hat{v}_n$. Then, the many-to-one matching induced by $A$ in the original two-sided matching instance is $\frac{1}{\opt^{1-\alpha}}$-Nash optimal, where $\opt$ is the maximum Nash welfare achieved by any matching.
\label{prop:Approximation_Transfer}
\end{restatable}
\begin{proof} (of \Cref{prop:Approximation_Transfer}) Let $A^*$ be the utilitarian optimal allocation in the fair division instance. Since $A$ is $\alpha$-utilitarian optimal, we have that
$$\W^\texttt{\textup{util}}(A) \geq \alpha \cdot \W^\texttt{\textup{util}}(A^*).$$
Let $\mu$ and $\mu^*$ denote the induced matchings for the allocations $A$ and $A^*$, respectively. From the welfare correspondence in \Cref{prop:Utilitarian_Nash_Relationship}, we know that
$$\W^\texttt{\textup{util}}(A) = (n+m) \log \W^\texttt{\textup{Nash}}(\mu), \text{ and}$$
$$\W^\texttt{\textup{util}}(A^*) = (n+m) \log \W^\texttt{\textup{Nash}}(\mu^*).$$
These equalities, together with the above inequality, imply that
$$\log \W^\texttt{\textup{Nash}}(\mu) \geq \alpha \cdot \log \W^\texttt{\textup{Nash}}(\mu^*),$$
or, equivalently,
$$\W^\texttt{\textup{Nash}}(\mu) \geq \opt^{\alpha},$$
where $\opt \coloneqq \W^\texttt{\textup{Nash}}(\mu^*)$ is the optimal Nash welfare. Thus, we have that
$$\W^\texttt{\textup{Nash}}(\mu) \geq \opt^{\alpha-1} \cdot \opt,$$
implying that $\mu$ is $\frac{1}{\opt^{1-\alpha}}$-Nash optimal.
\end{proof}

\paragraph{Matroid constraints.} We will now impose cardinality constraints on allocations in the fair division instance. These constraints will ensure that the induced matching satisfies capacity constraints. Specifically, for any agent $i \in [n]$, define its \emph{feasibility set} $\F_i \coloneqq \{S \subseteq [m]: |S| \leq c_i\}$ as the set of all subsets of at most $c_i$ items, where $c_i$ is the capacity of the corresponding firm $f_i$.

Define the ground set $X \coloneqq \{(i,S): i \in [n] \text{ and } S \subseteq \F_i\}$ as the set of all feasible agent-bundle pairs. Call any collection of such pairs $\mathcal{S} \subseteq X$ \emph{independent} if, for any $i \in [n]$, $\mathcal{S}$ contains at most pair of the form $(i,S)$ for some $S \in \F_i$. Let $\mathcal{J}$ denote the set of all independent subsets of the ground set. It is easy to check that the set system $\mathcal{M} = (X,\mathcal{J})$ is a matroid~\citep{O22matroid}. 

\citet{FNW78analysis} showed that a natural greedy algorithm gives a $\frac{1}{2}$-approximation algorithm for maximizing any monotone submodular function under arbitrary matroid constraints. At each step, the algorithm adds a feasible element that gives the highest marginal improvement in the objective. The algorithm uses only a polynomial number of value queries.
%
%

\begin{restatable}[\citealp{FNW78analysis}; Theorem 2.1]{proposition}{}
There is an algorithm that, given any fair division instance with monotone submodular valuations and arbitrary matroid constraints, makes a polynomial number of value queries and returns an allocation that is $\frac{1}{2}$-utilitarian optimal.
\label{prop:Submodular-Utilitarian-Approximation}
\end{restatable}

We are now ready to prove our main result (\Cref{thm:Nash-Approx-OPT-Dependent}) which gives an $\frac{1}{\sqrt{\opt}}$-approximation algorithm for maximizing Nash welfare in the two-sided matching problem.

\begin{proof} (of \Cref{thm:Nash-Approx-OPT-Dependent})
Given any two-sided matching instance $\I = \langle W,F,C, \V \rangle$, we will construct the corresponding fair division $\I' = \langle [n], [m], \hat{V} \rangle$ with modified valuation functions (\Cref{def:Modified_Valuation}).

From \Cref{prop:Modified_Valuation_Submodular}, we know that the modified valuations are monotone and submodular. Hence, by applying \Cref{prop:Submodular-Utilitarian-Approximation}, we obtain an algorithm for computing a $\frac{1}{2}$-utilitarian optimal allocation in the fair division instance. This algorithm uses a polynomial number of value queries. Note that any such value query in the fair division problem can be simulated in polynomial time. Thus, the algorithm runs in polynomial time.

The allocation constructed by this algorithm assigns each item to at most one agent. Furthermore, due to the independent set constraint of the matroid, an agent never receives more items than its (corresponding firm's) capacity. Finally, observe that no item is left unassigned; indeed, by the nonzeroness assumption for Nash welfare~(\Cref{sec:Preliminaries}), there exists some matching in the two-sided instance that satisfies capacity constraints and has nonzero Nash welfare. This, in turn, implies that there exists an allocation with nonzero utilitarian welfare satisfying the independent set constraints.

From \Cref{prop:Approximation_Transfer}, we know that the corresponding induced matching is $\frac{1}{\sqrt{\opt}}$-Nash optimal, as desired.
\end{proof}


\subsection{Proof of \Cref{thm:QPTAS-Constant-Firms}}


\QPTASalgorithm*

\begin{proof} (of \Cref{thm:QPTAS-Constant-Firms}) The algorithm uses the idea of \emph{bucketing} (or discretization) to classify the workers and firms according to the \emph{range} in which they value each other.

Formally, for any given $\eps > 0$ and any firm $f$, define a set of $\tau+1$ \emph{buckets} $B_0^f,B_1^f,B_2^f,\dots,B_{\tau}^f$, where $\tau \coloneqq \lceil \log_{1+\eps} v_{\max} \rceil$ and $v_{\max}$ is the maximum valuation of any agent for any other agent, i.e., $v_{\max} \coloneqq \max \{ \{v_{w,f}\}_{ (w,f) \in W \times F} \cup \{v_{f,w}\}_{ (w,f) \in W \times F} \}$. For any $i \in [\tau]$, the bucket $B_i^f$ denotes the set of workers who value firm $f$ between $(1+\eps)^{i-1}$ and $(1+\eps)^{i}$. Specifically, for $i \in [\tau-1]$, 
$$B_i^f \coloneqq \{w \in W: (1+\eps)^{i-1} \leq v_{w,f} < (1+\eps)^{i} \}, \text{ and }$$
$$B_\tau^f \coloneqq \{w \in W: (1+\eps)^{\tau-1} \leq v_{w,f} \leq (1+\eps)^{\tau} \}.$$
Additionally, the bucket $B_0^f$ contains the workers who value firm $f$ at $0$, i.e., 
$$B_0^f \coloneqq \{w \in W: v_{w,f} = 0 \}.$$
For each bucket $B_i^f$, we further define $\tau+1$ \emph{sub-buckets} $b_0^{f,i},b_1^{f,i},\dots,b_\tau^{f,i}$, where, for any $j \in [\tau]$, the sub-bucket $b_j^{f,i}$ denotes the set of workers in the bucket $B_i^f$ whom the firm $f$ values between $(1+\eps)^{j-1}$ and $(1+\eps)^{j}$. 
That is, for $j \in [\tau-1]$, 
$$b_j^{f,i} \coloneqq \{w \in B_i^f: (1+\eps)^{j-1} \leq v_{f,w} < (1+\eps)^{j} \}, \text{ and }$$
$$b_\tau^{f,i} \coloneqq \{w \in B_i^f: (1+\eps)^{\tau-1} \leq v_{f,w} \leq (1+\eps)^{\tau} \}.$$
The workers in $B_i^f$  valued at $0$ are assigned to $b_0^{f,i}$, i.e.,
$$b_0^{f,i} \coloneqq \{w \in B_i^f: v_{f,w} = 0 \}.$$
Thus, membership in a sub-bucket specifies the valuations of a worker and a firm for each other within a multiplicative factor of $(1+\eps)$ in case of nonzero values, and rounded up to $1$ in case of zero values.

Our algorithm \emph{guesses} the number of workers in each sub-bucket in an optimal solution. Each such guess, if feasible for the given capacities, induces a matching. The Nash welfare of this matching can be correctly computed to within a multiplicative factor of $(1+\eps)$ by only knowing the \emph{number} of workers in each sub-bucket. Importantly, the algorithm does not need to know that \emph{identities} of the workers.

The total number of sub-buckets across all firms 
is $n \cdot (\tau+1)^2$, or $\OO(n \tau^2)$. For $m$ workers, the total number of guesses made by the algorithm is at most $m^{\OO(n \tau^2)}$. To see this, note that the number of guesses made by the algorithm is equivalent to distributing $m$ indistinguishable balls into $n (\tau+1)^2$ bins.

For each guess, the algorithm checks its feasibility (i.e., whether it satisfies the capacity constraints) and computes its Nash welfare. As noted above, the Nash welfare can be computed correctly to within a multiplicative factor of $(1+\eps)$. Thus, for each guess, the feasibility check and Nash welfare computation can be performed in polynomial time. The guess with the maximum Nash welfare is returned as the solution. 

When the valuations are polynomially-bounded, we have that $\tau = \lceil \log_{1+\eps} v_{\max} \rceil = \OO(\log_{1+\eps} mn)$. Thus, the running time of the algorithm is $\OO(m^{\OO \left( n \cdot \log_{1+\eps}^2 mn \right) })$. By noting that the number of firms $n$ is constant and $\log (1+\eps) > \eps$, we obtain the desired implication.
\end{proof}

We will now discuss the proof of \Cref{cor:FPT-Exact-Constant-Firms}.

\ExactBucketing*

\begin{proof} (of \Cref{cor:FPT-Exact-Constant-Firms}) The proof is very similar to that of \Cref{thm:QPTAS-Constant-Firms}. The main difference is that the buckets and sub-buckets are now defined on the \emph{exact} valuations instead of in terms of powers of $(1+\eps)$.

Specifically, let $c_0$ be a constant such that $v_{\max} \leq c_0$. Then, for each firm $f$, we define $c_0+1$ buckets $B_0^f,B_1^f,B_2^f,\dots,B_{c_0}^f$ as follows: For $i \in \{0,1,2,\dots,c_0\}$, we have
$$B_i^f \coloneqq \{w \in W: v_{w,f} = i \}.$$
For each bucket $B_i^f$, we define $c_0+1$ \emph{sub-buckets} $b_0^{f,i},b_1^{f,i},\dots,b_\tau^{f,i}$, as follows: For $j \in \{0,1,2,\dots,c_0\}$, we have
$$b_j^{f,i} \coloneqq \{w \in B_i^f: v_{f,w} = i \}.$$
Note that by knowing the number of workers in each sub-bucket, the algorithm can compute the Nash welfare of the induced matching \emph{exactly}.

The total number of sub-buckets across all firms is $n \cdot (c_0+1)^2$. When the number of firms $n$ is constant, the number of sub-buckets is also constant. Thus, by a similar analysis as in the proof of \Cref{thm:QPTAS-Constant-Firms}, we can observe that the algorithm runs in polynomial time.
\end{proof}

\section{Omitted Proofs from \Cref{sec:Parameterized Algorithms}}

\subsection{Proof of \Cref{thm:3m}}

\DPalgorithm*
\begin{proof}
  Let $\langle W,F,C, \V \rangle$ be a given matching instance. We design a dynamic programming algorithm. Let $f_1,\dots,f_n$ be an arbitrary ordering of firms. We define the table $T$ as follows: For every $i\in [n]$ and $S\subseteq W$, 
\begin{equation*}
    \begin{split}
         T[i,S]= & \text{ maximum Nash product when the first $i$ firms are} \\
         & \text{ matched to the workers in the set } S.
    \end{split} 
\end{equation*}

We compute the entries of the table $T$ as follows:
For each subset $S\subseteq W$, 
\begin{equation}\label{eq:fpt-m-eq1}
    \begin{split}
    T[1,S] = \begin{cases} \W_{f_1}(S) & \text{ if } |S|\leq c_1, \\
    0 & \text{ otherwise}
    \end{cases}
    \end{split}
\end{equation}
Further, for any $i>1$ and any subset $S\subseteq W$,
\begin{equation}\label{eq:fpt-m-eq2}
    T[i,S] = \max_{S'\subseteq S, |S'|\leq c_i}\{\W_{f_i}(S') \times T[i-1,S\setminus S'] \}.
\end{equation}


To prove the correctness of the algorithm, we will need the following lemma. 

\begin{restatable}[]{lemma}{DPcorrectness}
    For every $i\in [m]$ and $S\subseteq W$, the maximum Nash product achievable in the subinstance consisting of the workers in $S$ and the firms $f_1,\dots,f_i$ is given by $T[i,S]$.
\label{lem:dp-rec-correctness}
\end{restatable}



\begin{proof} 
    We prove it by induction on $i$. 
    
    {\bf Base Case:} When $i=1$ and $|S|\leq c_1$, all the workers are assigned to $f_1$ and the Nash social welfare is $\W_{f_1}(S)$. If $|S|> c_1$, then there is at least one worker that is not assigned to $f_1$, thus Nash social welfare is $0$. Thus, the base case is correct.

    {\bf Induction Step:} Suppose that the table entries are computed correctly for all $i<j$ and all $S\subseteq W$ using \Cref{eq:fpt-m-eq1,eq:fpt-m-eq2}. We next prove it for $i=j$. Let $\mu$ be a matching that maximizes the Nash product when first $j$ firms are matched to the workers in the set $S$.  Let $\hat{S}=\arg \max_{S'\subseteq S, |S'|\leq c_j}\{\W_{f_j}(S') \times T[j-1,S\setminus S']$. Let $\hat{\mu}$ be a matching that matches first $j-1$ firms to the workers in the set $S\setminus \hat{S}$ and Nash product of $\hat{\mu}$ is same as $T[j-1,S\setminus \hat{S}]$. Clearly, such a matching exists and it is an optimal matching due to the induction. Let us construct a matching $\mu^\star$ as follows: for every agent $a\in \{f_1\cup \ldots \cup f_{j-1}\}\cup (S\setminus \hat{S})$, $\mu^\star(a)=\hat{\mu}(a)$, and $\mu^\star(f_j)=\hat{S}$. That is, for the first $j-1$ firms and the workers in $S\setminus \hat{S}$, both the matchings are same, and $\hat{S}$ is assigned to the $j$th firm. The Nash product of $\mu^\star$ is $\W_{f_j}(\hat{S})\times T[j-1,S\setminus \hat{S}]$. Clearly, the Nash product of $\mu^\star$ is at most the Nash product of $\mu$. We claim that these two values are same. Towards this, we show that the Nash product of $\mu^\star$ is at least the Nash product of $\mu$, that will lead to the equality. Let $S'=\mu(f_j)$. Consider a matching  $\mu'$ that is a restriction of $\mu$ on first $j-1$ firms and the workers in the set $S\setminus S'$. That is, for every $a\in \{f_1\cup \ldots \cup f_{j-1}\}\cup (S\setminus S')$, $\mu'(a)=\mu(a)$. Then, $\mu'$ is a valid candidate for the entry in $T[j-1,S\setminus S']$, i.e., $T[j-1,S\setminus S']$ is at least the Nash product of $\mu'$. Thus, the Nash product of $\mu$ is at most $\W_{f_j}(S')\times T[j-1,S\setminus S']$. Due to the choice of $\hat{S}$, $\W_{f_j}(S')\times T[j-1,S\setminus S']\leq \W_{f_j}(\hat{S})\times T[j-1,S\setminus \hat{S}]$. Thus, the Nash product of $\mu$ is at most the Nash product of $\mu^\star$. Hence, the Nash product of $\mu$ and $\mu^\star$ is same. Thus, $T[j,S]$ is computed correctly.
\end{proof}
Due to~\Cref{lem:dp-rec-correctness}, $T[n,W]$ is computed correctly and we can obtain the matching by 
backtracking.
Since we compute the table entry for all subsets $S$ of $W$, and to compute each entry we check all subsets of $S$ in the worst case, the running time is upper bounded by $n\sum_{i=0}^m \binom{m}{i}2^i$, or, $n \cdot 3^m$. 
\end{proof}

\subsection{Proof of \Cref{thm:param-m-fptas}}

\PolyMultThm*

Before dwelling into the proof, let us introduce some relevant definitions. The {\em characteristic vector} of a subset $S\subseteq [m]$, denoted by $\chi(S)$, is an $m$-length binary string whose $i^{\textup{th}}$ bit is $1$ if and only if  $i \in S$.  Two binary strings of length $m$ are said to be \emph{disjoint} if for each $i\in [m]$, the $i^\textup{th}$ bits in the two strings are different. The {\em Hamming weight} of a binary string $S$, denoted by ${\cal H}(S)$, is the number of $1$s in the string $S$. A monomial $y^i$ is said to have Hamming weight $w$ if the degree term $i$, when represented as a binary string, has Hamming weight $w$. The {\em Hamming projection} of a polynomial $p(y)$ to $h$, denoted by ${\cal H}_{h}(p(y))$, is the sum of all the monomials of $p(y)$ which have Hamming weight $h$. A polynomial $p(y)$ \emph{contains a monomial} $y^i$ if the coefficient of $y^{i}$ is nonzero. 
The {\em representative polynomial} of $p(y)$, denoted by ${\cal R}(p(y))$, is a polynomial derived from $p(y)$ by changing the coefficients of all monomials contained in it to $1$. 

 We begin with the following known result. 





\HWDisjointness*

Next, we prove our result.
\begin{proof} (of \Cref{thm:param-m-fptas})
   Let $\I = \langle W,F,C, \V \rangle$ be the given instance, where $F=\{f_1,\dots,f_n\}$ and $W=\{1,\dots,m\}$ are the sets of firms and workers, respectively. 
   Let $\eta$ be the maximum possible Nash product for $\I$; thus, $\eta \leq (mv_{\max})^{m+n}$. Let $Z \coloneqq \{1,(1+\eps), (1+\eps)^2,\ldots,(1+\eps)^q,(1+\eps)^{q+1}\}$, where $q$ is the largest positive integer such that $(1+\eps)^q \leq \eta$. For every $j\in [n]$, $s\in [m]$, and $\ell \in Z$, we will construct a polynomial $p_{s,\ell}^j$ in which every nonzero monomial corresponds to an assignment of $s$ workers to the first $j$ firms $f_1,\dots,f_j$ such that the Nash product is at least $\ell$.   
   
   We will construct these polynomials $p_{s,\ell}^j$ iteratively. First, we will construct a polynomial $h_{s,\ell}^j$ in which every nonzero monomial corresponds to an assignment of some set $X \subseteq W$ of $s$ workers to the $j^\textup{th}$ firm $f_j$ such that $\W_{f_j}(X)\geq \ell$. 
   That is, for every $j\in [n]$, $s\in [c_j]$, and $\ell \in Z$,
   \begin{equation}
       h^j_{s,\ell} \coloneqq \sum_{\substack{X\subseteq W, |X|=s, \\ {\cal W}_{f_j}(X)\geq \ell}} y^{\chi(X)}.
   \tag{\ref{eq:type1}}
   \end{equation}

   Define $p^1_{s,\ell} \coloneqq h^1_{s,\ell}$. 
 %
%
 For every $j\in \{2,\ldots,n\}, s\in [m]$ and $\ell \in Z$, define
 \begin{equation}
       p^j_{s,\ell} \coloneqq {\cal R}\Big({\cal H}_s\Big(\sum_{\substack{s=s'+s'', s'\leq c_j, \\ \ell=\ell'\times \ell'', \ell',\ell''\in Z}} h_{s',\ell'}^j\times p_{s'',\ell''}^{j-1}\Big)\Big),
 \tag{\ref{eq:typej}}
 \end{equation}
where ${\cal R}(\cdot)$ is the representative polynomial and ${\cal H}_s(\cdot)$ is the Hamming projection of weight $s$. 
The ${\cal H}(\cdot)$ operator ensures disjointness due to~\Cref{cor:HW-disjointness}. The ${\cal R}(\cdot)$ operator is only required for the running time. 

We return the largest value of $\ell$ for which $p_{m,\ell}^n$ is nonzero. The corresponding matching can be found by backtracking. 

To argue correctness, we will prove that if the algorithm returns~$\ell^\star$, then $\eta \leq (1+\eps)^{n+1}\ell^\star$. Towards this, we show that if $(1+\eps)^q\leq \eta \leq (1+\eps)^{q+1}$, then $p^n_{m,\ell}$ is nonzero for some $\ell \in \{(1+\eps)^{q-n},\ldots, (1+\eps)^{q+1}\}$. Thus, we return a matching $\tilde{\mu}$ with Nash product at least $(1+\eps)^{q-n}$. Since, $\eta \leq (1+\eps)^{q+1}$, it follows that $\eta \leq (1+\eps)^{n+1} \ell^\star$. Thus, $\W^\texttt{\textup{Nash}}(\mu) \leq (1+\eps)^{\frac{n+1}{m+n}}\W^\texttt{\textup{Nash}}(\tilde{\mu})$, where $\mu$ is a Nash optimal matching. 
The following lemma helps to conclude it.     

\begin{lemma}\label{lem:fptas-correctness}
    Let $\mu$ be a matching that maximises Nash product.  For each $j\in [n]$, let $\eta_j=\prod_{i\in [j]}{\cal W}_{f_i}(\mu(f_i))$ and $\alpha_j=\sum_{i\in [j]}|\mu(f_i)|$. Let $(1+\epsilon)^{q_j} \leq \eta_j \leq (1+\epsilon)^{q_j+1}$, where $j\in [n]$. Then, for each $j\in [n]$ and for an $\ell \in \{(1+\epsilon)^{q_j-j},\ldots, (1+\epsilon)^{q_j}\}$, $p^{j}_{\alpha_j,\ell}$ contains the monomial $y^{\chi(\mu(f_1)\cup \ldots \cup \mu(f_j))}$.
\end{lemma}

\begin{proof}
    We prove it by induction on $j$.

    {\bf Base Case:} $j=1$. Let $(1+\eps)^{q_1} \leq \eta_1 \leq (1+\eps)^{q_1+1}$. Clearly, $|\mu(f_1)|\leq c_1$. Thus, we consider $s=|\mu(f_1)|$ and $\ell = (1+\eps)^{q_1}$ in \Cref{eq:type1}. Since ${\cal W}_{f_1}(\mu(f_1))\geq (1+\eps)^{q_1}$, $p_{\alpha_1,(1+\eps)^{q_1}}^1$ contains the monomial $y^{\chi(\mu(f_1))}$.

    {\bf Induction Step:} Suppose that the claim is true for $j=j'-1$. We next prove it for $j=j'$. Let $(1+\eps)^{q_{j'}} \leq \eta_{j'} < (1+\eps)^{q_{j'}+1}$ and $(1+\eps)^{q_{j'-1}} \leq \eta_{j'-1} < (1+\eps)^{q_{j'-1}+1}$. Note that ${\cal W}_{j'}(\mu(f_{j'}))=\frac{\eta_{j'}}{\eta_{j'-1}}> \frac{(1+\eps)^{q_{j'}}}{(1+\eps)^{q_{j'-1}+1}}=(1+\eps)^{q_{j'}-q_{j'-1}-1}$. Furthermore, ${\cal W}_{j'}(\mu(f_{j'})) \geq 1$ as Nash product of $\mu$ is nonzero.  Since $(1+\eps)^{q_{j'-1}} \leq \eta_{j'-1} < (1+\eps)^{q_{j'-1}+1}$, due to induction hypothesis, for some $\ell'' \in \{(1+\eps)^{q_{j'-1}-(j'-1)},\ldots, (1+\eps)^{q_{j'-1}}\}$,  $p_{\alpha_{j'-1},\ell''}^{j'-1}$ contains a monomial $y^{\chi(\mu(f_1)\cup \ldots \cup \mu(f_{j'-1}))}$. Consider $\ell'=\max\{1,(1+\eps)^{q_{j'}-q_{j'-1}-1}\}$. Let $\ell=\ell'\times \ell''$, $s=\alpha_{j'}, s'=|\mu(f_{j'})|$, and $s''=\alpha_{j'-1}$. While constructing polynomial $p_{s,\ell}^{j'}$, we will consider above $s',s'',\ell'$, and $\ell''$.  Thus, we consider the multiplication of $h_{s',\ell'}^{j'}$ and  $p_{s'',\ell''}^{j'-1}$ in the construction of the polynomial $p_{s,\ell}^{j'}$. Since ${\cal W}_{j'}(\mu(f_{j'}))\geq \ell'$, $h_{s',\ell'}^j$ is nonzero. Furthermore, $p_{s'',\ell''}^{j'-1}$ contains the monomial $y^{\chi(\mu(f_1)\cup \ldots \cup \mu(f_{j'-1}))}$, Thus, we considered the multiplication of $y^{\chi(\mu(f_1)\cup \ldots \cup \mu(f_{j'-1}))}$ and $y^{\mu(f_{j'})}$. Since $\mu$ is a matching, $\mu(f_{j'})$ is disjoint from $\mu(f_1)\cup \ldots \cup \mu(f_{j'-1})$. Hence, $p_{s,\ell}^{j'}$ contains the monomial $y^{\chi(\mu(f_1)\cup \ldots \cup \mu(f_{j'}))}$. If $\ell'=(1+\eps)^{q_{j'}-q_{j'-1}-1}$ and $\ell''\geq (1+\eps)^{q_{j'-1}-(j'-1)}$, then $\ell \geq (1+\eps)^{q_{j'}-j'}$. If  $(1+\eps)^{q_{j'}-q_{j'-1}-1}<1$, $q_{j'}=q_{j'-1}$ as $q_{j'}$ cannot be smaller than $q_{j'-1}$. In this case, $\ell'=1$. Thus, $\ell \geq (1+\eps)^{q_{j'-1}-(j'-1)}=(1+\eps)^{q_{j'}-j'+1}$. This completes the proof.
\end{proof}

Since $|Z|=\OO(\log_{1+\eps} \eta)$, we compute $\OO(nm\log_{1+\eps} \eta)$ polynomials. Note that the degree of the polynomial 
is at most $2^m$ (as the the $m$-length binary vector in the polynomial can have all $1$s). Each polynomial can be computed in $\OO(m^2 \cdot 2^m \cdot \log \eta)$ time due to the following result and the fact that $s\leq m$ and $\ell \in Z$. 
\Moenck*
Since $\eta \leq (mv_{\max})^{m+n}$, $\log_{1+\eps} \eta \leq (m+n)\log_{1+\eps}(mv_{\max})$. By changing the base of logarithm from $(1+\eps)$ to $2$, and the fact that $\log(1+\eps)>\nicefrac{\eps^2}{2}$, for every $\eps \in (0,1]$, we get the running time of $\OO^\star(\nicefrac{1}{\eps^4} \cdot 2^m)$.
\end{proof}

\section{Omitted Proofs from \Cref{sec:Restricted_Domains}}

\subsection{Proof of \Cref{thm:symm-Binary-Vals}}

\SymmBinaryVals*

\begin{proof}
   The algorithm and analysis, both are similar to the one for binary valuations and one-sided preferences~\citep{BKV18greedy}. 
   
   Let $\langle W,F,C,\V \rangle$ be the given instance of the matching problem. In case of one-sided preferences, a worker can be assigned to any firm. While in case of two-sided preferences, a worker can only be assigned to those firms that it values at $1$. Furthermore, firms have capacities in our problem. Thus, we need to \emph{carefully} choose our initial suboptimal matching. We use \Cref{prop:Zero_Nash} to generate a matching with nonzero Nash welfare. Given a suboptimal matching $\mu$, we greedily update it to obtain another matching whose Nash welfare is at least the Nash welfare of $\mu$. We will show that if the Nash welfare of both the matchings is same, then $\mu$ itself is an optimal matching. Towards this, we first construct the following directed graph $G(\mu)$. 
        \begin{itemize}
            \item For every firm $f\in F$, there is a vertex in  $G(\mu)$. 
            \item For a firm $f\in F$, let $W^+(f)$ be the set of workers valued at $1$ by $f$. For a pair of firms $f,f'$, we add $|\mu(f)\cap W^+(f')|$ directed arcs from $f$ to $f'$. 
        \end{itemize}
   Note that a worker $w$ belongs to the set $\mu(f)\cap W^+(f')$ if $w$ is currently assigned to $f$ and is valued by $f'$ at $1$. Due to symmetric valuations $w$ also values $f'$ at $1$.
   
    Let $P=(f_1,\ldots,f_k)$ be a simple directed path in $G(\mu)$. Then, we obtain a new matching $\mu_P$ by rematching firms in $P$ as follows: let $w$ be a worker in $\mu(f_i)\cap W^+(f_{i+1})$, where $i\in [k-1]$. Rematch $w$ to $f_{i+1}$, that is $\mu_P(w)=f_{i+1}$, $\mu_P(f_i)=\mu(f_i)\setminus \{w\}$, $\mu_P(f_{i+1})=\mu(f_{i+1})\cup \{w\}$. Rest of the matching is same as $\mu$. We call a path $P$ \emph{good} if $\mu_P$ is a feasible matching, i.e., all capacity constraints are met, and the Nash welfare of $\mu_P$ is more than the Nash welfare of $\mu$. Let $P_{uv}$ be the set of all good paths from $u$ to $v$ in $G(\mu)$. Let ${\cal P}=\{P_{uv}\colon u,v\in G(\mu)\}$. Let $P^\star = \arg\max_{P\in {\cal P}}\W^\texttt{\textup{Nash}}(\mu(P))$. Then, we update $\mu$ to $\mu_{P^\star}$.

The following lemma shows the progress towards the optimal solution made by this rematching. 

\begin{lemma}
    Let $\mu^\star$ be a Nash optimal matching. Given a suboptimal matching $\mu$, there exist firms $f$ and $f'$ such that there is a directed path $P$ from $f$ to $f'$ in $\mu$ such that 
    \begin{equation*}
        \begin{split}
            \ln \W^\texttt{\textup{Nash}}(\mu^\star) - \ln \W^\texttt{\textup{Nash}}(\mu(P)) \leq \\
            \Big(1-\frac{1}{m}\Big) (\ln \W^\texttt{\textup{Nash}}(\mu^\star) - \ln \W^\texttt{\textup{Nash}}(\mu))
        \end{split}
    \end{equation*}
\end{lemma}

   \begin{proof}
       The proof is same as the proof of Lemma 5.1 in~\citep{BKV18greedy}. The agent refers to firm and good refers to worker.
   \end{proof}



Note that Remark 1 in~\citep{BKV18greedy} is true for our case as well. Thus, every path in $P_{uv}$ is equally good for us. Furthermore, note that any rematching only affects the valuation of firms, for workers, it is always $1$. Furthermore, the effect on the valuation of firms is also same as the one in two-sided preferences. This is the reason that our analysis is exactly the same. 

\begin{lemma}
    The algorithm runs in $2m(n+1)\ln(nm)$ time.
\end{lemma}

\begin{proof}
    The proof is same as the proof of Theorem 3.2 in \citep{BKV18greedy}.
\end{proof}
\end{proof}

\subsection{Proof of \Cref{thm:Combined-Restricted-Domains}}

\Combined*

\begin{proof} [Proof of \Cref{thm:Combined-Restricted-Domains} (1)]
     We begin with proving the first result. Since the matching instance has degree at most two, the graph is a disjoint union of paths and cycles. Note that we can compute the matching for each component separately. Thus, here we discuss the algorithm to compute the matching when the given instance is a path or a cycle. Suppose that the matching instance is an odd length path $P=(f_1, w_1,\ldots,f_\ell,w_\ell)$, i.e., the number of edges in $P$ is odd. Note that $\{f_iw_i \colon i\in [\ell]\}$ is the unique matching that has nonzero Nash welfare. Thus, it is optimal. If $P$ is an even length path that start from a vertex in $F$, then the Nash product is zero as the number of firms is more than the number of workers. Suppose that  $P=(w_1,f_1,\ldots,w_\ell)$ is an even length path that start from a vertex in $W$. Then, $|W|=|F|+1$. Thus, there is only one firm who is matched to two workers. We guess this firm, say $f_i$, $i\in [\ell-1]$. Then, $f_i$ is matched to its both the neighbors, $w_{i}$ and $w_{i+1}$. Now, consider the paths $P_1=(w_1,\ldots,f_{i-1})$ (if $i\neq 1$) and $P_2=(f_{i+1},\ldots,w_\ell)$ (if $i\neq \ell-1$). Note that we can find the matching in $P_1$ and $P_2$ as discussed earlier. Suppose that $P=(f_1,w_1,\ldots,f_\ell,w_\ell)$ is a cycle. Then, the number of firms is same as the number of workers. Thus, every firm is matched to only one worker. We first guess a matching edge say $f_iw_i$, where $i\in[\ell]$. Then, the problem is reduced to find a matching in the path $P'= P\setminus \{f_i,w_i\}$ of odd length. The correctness follows due to the correctness of finding matching in the path of odd length.
\end{proof}

\begin{proof} [Proof of \Cref{thm:Combined-Restricted-Domains} (2)]
Let $\langle W,F,C, \V \rangle$ be a given matching instance. We first observe that if there exist two firms $f,f' \in F$ such that $N(f)\cap N(f')=\{w,w',w''\}$, i.e., their neighborhood is same, then it is a no-instance of the problem as $f$ and $f'$ both cannot be matched to two workers. Thus, from now onward, we assume that no two firms have the same neighborhood. Furthermore, note that for a yes-instance $|W|=2|F|$. Next, using the following reduction rule, we reduce our instance to another instance of the same problem where two firms can only have one common worker in their neighborhood. 
\begin{reduction rule}\label{rr:degree3cap2}
Let $\langle W,F,C, \V \rangle$ be a given matching instance. Suppose that $f,f'$ be two firms in $F$ such that  $N(f)\cap N(f')=\{w_1,w_2\}$. Let $w_3\in N(f)\setminus N(f')$ and $w_4\in N(f')\setminus N(f)$. Delete $f,f'$ from $F$ and $w_i$, $i\in [4]$, from $W$.
\end{reduction rule}

\begin{lemma}
    Reduction Rule~\ref{rr:degree3cap2} is correct.
\end{lemma}

\begin{proof}
    Let $\langle W',F',C', \V' \rangle$ be the matching instance obtained after applying  Reduction Rule~\ref{rr:degree3cap2}. Let 
    $\mu$ be an optimal matching to $\langle W,F,C, \V \rangle$. Consider a matching $\mu'$, where $\mu'(a)=\mu(a)\setminus (W'\cup F')$, for all $a\in W'\cup F'$. If $\mu(a)\setminus (W'\cup F')$, then $a$ is unmatched in $\mu'$. We show that $\mu'$ is an optimal solution to $\langle W',F',C', \V' \rangle$. Towards this, we first claim that $\mu(f)\cup \mu(f')=\{w_1,w_2,w_3,w_4\}$. If $\mu(f) = \{w_1,w_2\}$, then $\mu(f')=w_4$ as degree of every firm is at most $3$ and $N(f')=\{w_1,w_2,w_4\}$. This contradicts that $\mu$ is a solution as $f'$ is not matched to two workers. Similarly, $\mu(f')\neq \{w_1,w_2\}$. Thus, $w_3\in \mu(f)$, $w_4\in \mu(f')$, $|\mu(f)\cap \{w_1,w_2\}|=1$, and  $|\mu(f')\cap \{w_1,w_2\}|=1$. Thus, $\mu(f)\cup \mu(f')=\{w_1,w_2,w_3,w_4\}$. Hence, for all $f\in F'$, $\mu'(f)=\mu(f)$, and for all $w\in W'$, $\mu'(w)=\mu(w)$. Suppose that there exist a matching $\mu^\star$ in the reduced instance $\langle W',F',C', \V' \rangle$ such that $\W^\texttt{\textup{Nash}}(\mu^\star)> \W^\texttt{\textup{Nash}}(\mu')$. Then, the Nash welfare of the matching $\hat{\mu}$, where $\hat{\mu}(a)=\mu^\star(a)$, for all agents in $F'\cup W'$, and $\hat{\mu}(a)=\mu(a)$, for all agents $a\in\{f,f',w_1,w_2,w_3,w_4\}$. Clearly, $\W^\texttt{\textup{Nash}}(\hat{\mu})> \W^\texttt{\textup{Nash}}(\mu)$, a contradiction.
    
    In the other direction, let $\mu$ be an optimal matching to $\langle W',F',C', \V' \rangle$. Consider the following constant-sized instance: $W''=\{w_1,w_2,w_3,w_4\}, F''=\{f,f'\}, c''_f=c_f, c''_{f'}=c_{f'}$, and for every agent $a\in W''\cup F''$, $v''_{a}=v_a$. Clearly, we can find an optimal matching $\mu''$ to $\langle W'',F'',C'',V''\rangle$ in polynomial time. We claim that $\mu'=\mu \cup \mu''$ is an optimal matching to $\langle W,F,C,V\rangle$. Suppose that $\mu^\star$ is an optimal matching to $\langle W,F,C,V\rangle$ and  $\W^\texttt{\textup{Nash}}(\mu^\star)> \W^\texttt{\textup{Nash}}(\mu')$. As argued above, $\mu^\star(f)\cup \mu^\star(f')=\{w_1,w_2,w_3,w_4\}$. Let $\hat{\mu}$ be a matching for $\langle W',F',C', \V' \rangle$, where $\hat{\mu}(a)=\mu^\star(a)$ for all $a\in W'\cup F'$. Clearly, $\W^\texttt{\textup{Nash}}(\hat{\mu})> \W^\texttt{\textup{Nash}}(\mu)$ as $\mu''$ is an optimal solution to $\langle W'',F'',C'',V''\rangle$, a contradiction.
\end{proof}

After the exhaustive application of Reduction Rule~\ref{rr:degree3cap2}, we know that a pair of firms can have only one worker in their common neighborhood. In other words, a pair of worker has only one common firm. Recall that we need to assign two workers to a firm. Thus, if we decide the pair of workers that need to assigned together to a firm, we obtain the desired matching.  
Towards this, we reduced the problem to the \WeightedMatching\ problem, in which given an edge-weighted graph $G$, the goal is to find a perfect matching of maximum weight\footnote{Here, we mean one-to-one matching}. Let $\langle W,F,C,V\rangle$ be the reduced instance of the problem. We create an instance of the \WeightedMatching\ problem as follows: $V(G)=W$, for every pair of workers $\{w,w'\}$, if they have a common firm in their neighborhood, then add an edge $ww'$ in $G$ and the weight of this edge is $\frac{1}{n+m}\log(\W_f(\{w,w'\}))$. Next, we find a perfect matching of maximum weight in $G$ in the polynomial time using the algorithm in~\citep{DBLP:conf/soda/Gabow90}. The correctness follows due to the following lemma. 

\begin{lemma}
    If there is no perfect matching in $G$, then $\langle W,F,C,V\rangle$ is a no-instance\footnote{Despite of the fact that given instance has nonzero Nash welfare, it is possible that there is no solution that matches every firm with two workers and has nonzero Nash welfare.} of Nash social welfare. Otherwise, the maximum Nash social welfare of the instance is same as the maximum weight of a perfect matching in $G$.
\end{lemma}

\begin{proof}
    We first prove the first claim. Suppose that $\mu$ is a solution to $\langle W,F,C,V\rangle$. Then, every firm $f\in F$ is matched to two workers in $\mu$. Let $\mu'=\{ww' \colon \{w,w'\} = \mu(f), f\in F\}$. Clearly, $\mu'$ is a perfect matching in $G$ as every worker is matched in $\mu$ and $|\mu(f)|=2$ for all $f\in F$. 

    We next prove the second claim. Let $\mu$ be a maximum weight perfect matching of $G$. Let $ww'$ be an edge in $\mu$. Then, we know that there exists a unique firm $f\in F$ that belongs to the common neighborhood of $w$ and $w'$. We create a matching $\mu'$ by assigning $w,w'$ to $f$. We claim that $\mu'$ is an optimal solution to $\langle W,F,C,V\rangle$. Suppose that there exists a matching $\mu^\star$ such that $\W^\texttt{\textup{Nash}}(\mu^\star)> \W^\texttt{\textup{Nash}}(\mu')$. Then, we create a perfect matching $\hat{\mu}$ for $G$ as described in the first claim. Since $(\prod_{f\in F}\W_f(\mu^\star(f)))^{\frac{1}{n+m}}>(\prod_{f\in F}\W_f(\mu^\star(f)))^{\frac{1}{n+m}}$, it follows that the weight of $\hat{\mu}$ is more than the weight of $\mu$, a contradiction.
\end{proof}
This completes the proof. 
\end{proof}

\begin{proof} [Proof of \Cref{thm:Combined-Restricted-Domains} (3)]
Since every worker values only one firm positively, it can only be assigned to this firm. Thus, we obtain optimal matching in the polynomial time.
\end{proof}

\begin{proof} [Proof of \Cref{thm:Combined-Restricted-Domains} (4)] 
The idea is similar to the bucketing algorithm for \Cref{thm:QPTAS-Constant-Firms}. Here, for every firm, we have bucket for each valuation, instead of $(1+\eps)$-sized bucket. Since the number of distinct valuations and the firms is constant, the total number of buckets is constant. Hence, the total number of guesses is constant and the algorithms runs in the polynomial time. The running time of the same algorithm is quasipolynomial when the number of distinct valuations is logarithmically bounded in the input size as now we have $m^{\OO(\log(n+m))}$ many guesses.
\end{proof}

\end{document}